\documentclass[british, a4paper]{article}

\usepackage{arxiv}

\usepackage{amsmath,amssymb,amsfonts,amsthm}
\usepackage{mathtools}
\usepackage{graphicx}       
\usepackage[utf8]{inputenc} 
\usepackage[T1]{fontenc}    
\usepackage{hyperref}       
\usepackage{url}            
\usepackage{booktabs}       
\usepackage{nicefrac}       
\usepackage{doi}
\usepackage{nth}
\usepackage{float}
\usepackage{caption}
\usepackage{subcaption}
\usepackage{adjustbox}
\usepackage[capitalize]{cleveref} 

\usepackage{pgfplots}
\DeclareUnicodeCharacter{2212}{−}
\usepgfplotslibrary{groupplots,dateplot}
\usetikzlibrary{patterns,shapes.arrows}
\pgfplotsset{compat=newest}

\DeclareMathAlphabet{\mymathbb}{U}{BOONDOX-ds}{m}{n}

\newcommand{\symi}{h}
\newcommand{\symk}{k}
\newcommand{\symt}{t}
\newcommand{\symkit}{k}
\newcommand{\A}{\bb{V}}
\newcommand{\imj}{j}
\newcommand{\R}{\mathbb{R}}
\newcommand{\C}{\mathbb{C}}
\newcommand{\bb}[1]{\mymathbb{#1}}

\newcommand{\mc}[1]{\mathcal{#1}}
\newcommand{\mt}[1]{\textrm{#1}}

\newcommand{\lt}{\left}
\newcommand{\rt}{\right}
\newcommand{\beeq}{\begin{equation}}
\newcommand{\eneq}{\end{equation}}
\newcommand{\matb}{\begin{matrix}}
\newcommand{\mate}{\end{matrix}}

\newcommand{\rvline}{\hspace*{-\arraycolsep}\vline\hspace*{-\arraycolsep}}
\newcommand{\norm}[1]{\left\lVert#1\right\rVert}

\DeclareMathOperator*{\argmin}{arg\,min}

\DeclareMathOperator*{\vect}{\mt{vec}}
\DeclareMathOperator*{\ve}{\mt{ve}}
\DeclareMathOperator*{\var}{Var}
\DeclareMathOperator*{\cov}{Cov}
\DeclareMathOperator*{\E}{E}

\newcounter{thms}

\newtheorem{corollary}[thms]{Corollary}
\newtheorem{lemma}[thms]{Lemma}
\newtheorem{assumption}[thms]{Assumption}

\theoremstyle{definition}

\theoremstyle{plain}
\newtheorem*{remark}{Remark}

\graphicspath{ {./pictures/} }

\title{Bayesian Error-in-Variables Models for the Identification of Power Networks}
\date{\nth{13} December 2021}

\author{ 
    \href{https://orcid.org/0000-0002-3474-1689}{\includegraphics[scale=0.06]{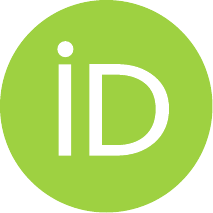}\hspace{1mm}Jean-Sébastien Brouillon} \\
	Institute of Mechanical Engineering,\\
	\'Ecole Polytechnique F\'ed\'erale de Lausanne (EPFL),\\
	CH-1015 Lausanne, Switzerland \\
	\texttt{jean-sebastien.brouillon@epfl.ch} \\
	\And
    \href{https://orcid.org/0000-0003-3357-4679}{\includegraphics[scale=0.06]{orcid.pdf}\hspace{1mm}Emanuele Fabbiani} \\
	Identification and Control of \\Dynamic Systems Laboratory,\\
	University of Pavia,\\
	Pavia, Italy \\
	\texttt{emanuele.fabbiani01@universitadipavia.it} \\
	\And
    \href{https://orcid.org/0000-0002-3505-7731}{\includegraphics[scale=0.06]{orcid.pdf}\hspace{1mm}Pulkit Nahata} \\
	Institute of Mechanical Engineering,\\
	\'Ecole Polytechnique F\'ed\'erale de Lausanne (EPFL),\\
	CH-1015 Lausanne, Switzerland \\
	\texttt{pulkit.nahata@epfl.ch} \\
	\And
    \href{https://orcid.org/0000-0003-2721-2161}{\includegraphics[scale=0.06]{orcid.pdf}\hspace{1mm}Keith Moffat} \\
	Department of Electrical Engineering \\and Computer Science,\\
	UC Berkeley,\\
	Berkeley, USA \\
	\texttt{keithm@berkeley.edu} \\
	\And
    \href{https://orcid.org/0000-0002-9649-5305}{\includegraphics[scale=0.06]{orcid.pdf}\hspace{1mm}Florian Dörfler} \\
	Automatic Control Laboratory,\\
	Swiss Federal Institute of Technology (ETH),\\
	Zurich, Switzerland \\
	\texttt{dorfler@control.ee.ethz.ch} \\
	\And
    \href{https://orcid.org/0000-0002-9492-9624}{\includegraphics[scale=0.06]{orcid.pdf}\hspace{1mm}Giancarlo Ferrari-Trecate} \\
	Institute of Mechanical Engineering,\\
	\'Ecole Polytechnique F\'ed\'erale de Lausanne (EPFL),\\
	CH-1015 Lausanne, Switzerland \\
	\texttt{giancarlo.ferraritrecate@epfl.ch} \\
}

\renewcommand{\headeright}{Technical Report}
\renewcommand{\undertitle}{Technical Report}
\renewcommand{\shorttitle}{J.S. Brouillon et al, Bayesian Error-in-Variable Models}

\hypersetup{
pdftitle={Bayesian Error-in-Variables Models for the Identification of Power Networks},
pdfauthor={Jean-Sébastien Brouillon, Emanuele Fabbiani, Pulkit Nahata, Florian Dörfler, Giancarlo Ferrari-Trecate},
pdfkeywords={Bayesian inference, Distribution grids, Error-in-variables, Line admittance estimation, Power systems identification},
}

\begin{document}
\maketitle

\begin{abstract}
	The increasing integration of intermittent renewable generation, especially at the distribution level, necessitates advanced planning and optimisation methodologies contingent on the knowledge of the grid, specifically the admittance matrix capturing the topology and line parameters of an electric network. However, a reliable estimate of the admittance matrix may either be missing or quickly become obsolete for temporally varying grids. In this work, we propose a data-driven identification method utilising voltage and current measurements collected from micro-PMUs. More precisely, we first present a maximum likelihood approach and then move towards a Bayesian framework, leveraging the principles of maximum a posteriori estimation. In contrast with most existing contributions, our approach not only factors in measurement noise on both voltage and current data, but is also capable of exploiting available a priori information such as sparsity patterns and known line parameters. Simulations conducted on benchmark cases demonstrate that, compared to other algorithms, our method can achieve significantly greater accuracy. 
\end{abstract}

\keywords{Bayesian inference \and Distribution grids \and Error-in-variables \and Line admittance estimation \and Power systems identification}

\section{Introduction} \label{section_intro}

A major key to realising green energy systems is the large-scale integration of renewable energy sources (RESs) in the distribution grid. Nevertheless, RES proliferation leads to additional risks such as reverse power flows and over-voltage---especially during periods of peak generation and low consumption \cite{Weng1}. Distribution grid operators are consequently required to put in place intelligent monitoring and control algorithms in order to maintain the existing levels of grid safety and reliability \cite{Schenato,LaBella2018,Iovine2019,Parisio2016}.  Deploying such algorithms efficiently requires the knowledge of the topology and the line parameters of the grid, embedded in its admittance matrix.

An exact estimate of the admittance matrix is hard to obtain for distribution grids, in particular as topological information and line parameter values either are unavailable for large chunks of the network, or become obsolete in the event of a topology change. To circumvent this issue, many recent contributions work out an up-to-date admittance matrix estimate by utilising data collected from micro-phasor measurement units ($\mu$PMUs). Although a more recent development than PMUs, which are commonly deployed in transmission systems, $\mu$PMUs have already been installed in distribution grids across America, Asia, and Europe, and their penetration is expected to steadily increase in the coming years \cite{kumar2020micro}. 

Due to its increasing relevance, the problem of identifying the topology and line parameters of a power grid has attracted considerable attention in the last few years. In \cite{cavraro2018graph, cavraro2019inverter}, an approach based on inverter probing is explored. Both works, besides employing approximate linearized power-flow equations, are restricted to radial networks. Albeit requiring voltage and current (or power) measurements at each bus of the grid, identification methods in \cite{du2019optimal, scaglione2017, yuan2016inverse, online_rls, tomlin} can be applied to both radial and meshed structures. In \cite{online_rls, tomlin}, structural properties of the admittance matrix, such as symmetry and Laplacianity are used to eliminate redundant admittance matrix parameters. Moreover, \cite{tomlin} proposes an adaptive Lasso algorithm promoting sparsity.

The tradeoff between voltage stability and the energy of the signal is the main challenge for the identification of power systems. To compensate for the lack of signal, online design-of-experiment procedures are presented in \cite{du2019optimal,scaglione2017}; nonetheless, the proposed algorithms require control authority on the state of the grid, and additional measurements of line power flows. They also neglect the structural and sparsity properties of the admittance matrix and stay limited by the small acceptable voltage variations.

All the foregoing works suffer from two limitations. First, they either completely disregard measurement errors or assume errors solely on certain measurements. This creates an estimation bias, for $\mu$PMUs introduce an measurement noise on all electric variables \cite{pinte2015low, sarri2016methods}. Second, they do not capitalize on grid information which may already be available \textit{a priori}, for instance sparsity patterns and known network sections and line parameters. To do away with the first limitation, \cite{patopa,wehenkel2020parameter} introduce error-in-variable (EIV) models taking into consideration all sources of measurement errors. That notwithstanding, they leave aside all prior information, including structural properties of the admittance matrix, which can potentially improve grid identification.

In this paper, we address the limitations of existing works by putting forth a novel Bayesian grid identification framework, which incorporates EIV models with an unbiased estimation of the error on both voltage and current data, and takes advantage of the principles of maximum likelihood estimation (MLE). Our approach exploits not only the inherent structural properties of the admittance matrix, but enables to exploit grid information known \textit{a priori}. In particular, we describe how to incorporate in the identification algorithm different pieces of information which may be available to grid operators, partly through data-driven Bayesian priors.

In order to substantiate the efficacy of our method, we conduct simulations on a large network with realistic voltage and current profiles, and $\mu$PMU noise levels compatible with the accuracy of actual commercial devices. We then compare the performance of our proposed methods with other grid identification procedures proposed in the literature. Our analysis shows not only that EIV models are needed to obtain reasonable grid estimates, but also that sparsity needs to be enforced if the topology of the network is unknown. Moreover, it substantiates the value of injecting prior information into the estimation algorithm, in case it is available.

The paper proceeds as follows: \cref{section_model,section_problem,subsec_noise_carac} define the identification problem, which we first solve using likelihood maximization in \cref{section_methods}. \cref{section_estimation,subsec_structural} introduce prior knowledge using the Bayesian framework. \cref{section_numeric} provides numerical methods for solving the optimization problem, \cref{section_results} presents a realistic simulation, which is further discussed in \cref{section_discussion}. \cref{section_conclu} concludes the paper and proposed future developments.

\subsection{Preliminaries and Notation}
\label{subsec:notation}
Let $\imj =\sqrt{-1}$ denote the imaginary unit. For $x \in \bb{C}^n$, $\overline{x}$ is its complex conjugate and $|x|$ its magnitude, both taken element-wise. The diagonal matrix of order $n$ associated with $x$ is denoted with $[x]$. The $\ell_1$ and $\ell_2$ norms of a vector $x$ are represented by $\norm{x}_1$ and $\norm{x}_2$, respectively. Throughout, $\bb{1}_n$ and $\bb{0}_n$ are $n$-dimensional vectors of all ones and zeros, whereas $\mathcal{I}_n$ and $\mathcal{O}_{n \times m}$ represent $n$-by-$n$ identity and $m$-by-$n$ zero matrices, respectively. The unit vector $e_i,~i= 1,..., n$ is the $i^{th}$ column of $\mathcal{I}_n$. For an $(m,n)$ matrix $A$, $A^\top$ denotes its transpose, $A_{i\cdot}$ its $i^{th}$ row vector, and $\vect(A) =[A_{1\cdot}^\top \cdots A_{n\cdot}^\top]^\top$ the $mn$-dimensional stacked column vector. Given a square matrix $A$, $\ve(A)$ is the $n(n-1)/2$-dimensional vector obtained by removing diagonal and supra-diagonal elements from $-\vect(A)$. The Kronecker product between matrices $A$ and $B$ is $A \otimes B$. Given $n$ elements $x_n$, $[x_i]_{i=1}^n$ is the stacked column vector $[x_{1}^\top \cdots x_{n}^\top]^\top$.

\textit{Random variables.} $\mathcal{X} \sim \mathcal{N}(\mu, \sigma^2)$ denotes a Gaussian random variable with expected value $\E[\mathcal{X}]=\mu$ and variance $\var[\mathcal{X}] = \sigma^2$.

\textit{Algebraic graph theory.} We denote by $\mathcal{G}(\mathcal{V},\mathcal{E},\mathcal{W})$ an undirected, weighted, and \textit{connected} graph, where $\mathcal{V}$ is the node set of cardinality $n$, $\mathcal{E} \subseteq (\mathcal{V}\times\mathcal{V})$ the edge set, and $\mathcal{W}$ collects the edge weights. The adjacency matrix $W \in \bb{C}^{n \times n}$ has elements $w_{\symi\symk}$ corresponding to the weight of the edges $(\symi,\symk) \in \mathcal{E}$ and zero otherwise. The matrix $L= [W\, \bb{1}_{n}]- W$, $L\in \bb{C}^{n \times n}$, is the Laplacian matrix associated with $\mathcal{G}$. By definition, a Laplacian matrix is symmetric and such that $L\bb{1}_n = \bb{0}_n$.

\section{Grid model and data collection} \label{section_model}
\subsection{Power grid model} \label{subsec_grid_model}
An electric distribution network is modeled as an undirected, weighted, and connected graph $\mathcal{G}(\mathcal{V},\mathcal{E},\mathcal{W})$, where the nodes in $\mathcal{V} = \{1, 2, \dots, n\}$ represent buses, either generators or loads, and edges represent power lines, each connecting two distinct buses and modeled after the standard lumped $\pi-$model \cite{wood2013power}. To each edge $(\symi, \symk) \in \mathcal{E}$ we associate a complex weight equal to the line admittance  $y_{\symi\symk}=g_{\symi\symk}+\imj b_{\symi\symk}$, where $g_{\symi\symk} > 0$ is the line conductance and $b_{\symi\symk} \in \R$ the line susceptance. 

The network is then completely represented by the $n$-by-$n$ complex admittance matrix $Y$, with elements $Y_{\symi\symk}=-y_{\symi\symk}$ for $\symi \neq \symk$ and $Y_{\symi\symk}=\sum_{\symi=1,\symi\neq \symk}^{n}y_{\symi\symk} +y_{s,\symi}$, where $y_{s,\symi}$ is the shunt element at the $\symi^{th}$ bus. The admittance matrix $Y$ is symmetric and typically sparse, as each bus is connected to few others: notably, this is the case in distribution grids, which are often characterized by a radial topology. Moreover, for network where shunt elements are negligible, $Y$ is Laplacian \cite{kundur2007power}.




We consider the network to be either single-phased or phase-balanced, and operating in sinusoidal regime. To each bus $h \in \mathcal{V}$, we associate a phasor voltage $v_{\symi}e^{\imj \theta_{\symi}}\in \C$, where $v_{\symi} > 0$ is the voltage magnitude and $\theta_{\symi} \in \R$ the voltage angle, and a phasor current $i_{\symi}e^{\imj \phi_{\symi}}\in \C$, representing the injection at the bus. We do not assume the presence of a point of common coupling (PCC), although one or more may be present as long as their fixed $v_0$ and $\theta_0$ are known. The current-voltage relation descending directly from Kirchhoff's and Ohm's laws is given by
\begin{equation}
	i = Yv,
	\label{eq:current-voltage}
\end{equation}
where $i \in \C^{n}$ is the vector of nodal current injections, and $v \in \C^{n}$ the vector of nodal voltages \cite{dorfler2018electrical}.

\subsection{Data collection} \label{subsec_data_collection}
We assume that each bus of interest is equipped with $\mu$PMU, while we do not require electrical variables to be measured on the lines.

\begin{assumption} \label{ass:observability}
The network is either completely observable, that is current injections and voltages are measured at each node, or a reduced network between the observed nodes is identified.
\end{assumption}

In transmission systems, where the reliability and the economic optimization of dispatch are primary concerns, 
synchronized phasor measurements are provided by PMUs, which sample the magnitude and phase of current and voltage phasors. Unfortunately, distribution networks are characterized by relatively small flows of active power and predominantly resistive lines, resulting in small phase differences between nodes. The accuracy of standard PMUs is not high enough to reliably appreciate such differences, making the devices ineffective.


Micro-syncrophasors improve the resolution and accuracy of PMUs by up to two orders of magnitude - see \cref{tab:pmu-error}, while preserving low costs. As PMUs, these devices sample the magnitude and phase of current and voltage: state-of-the-art models achieve frequency up to 120 Hz \cite{von2017precision}. While the fast sampling rate generates a large number of data points in a short time span, the samples are highly correlated. Indeed, the characteristic frequencies of load variations are much lower than the sampling rate of $\mu$PMUs. Moreover, due to structure of the network and the low phase difference between nodes, samples collected on different buses are highly correlated with each other \cite{tomlin}.

\begin{table}
	\centering
	\begin{tabular}{l c c c}
		\hline 
		Metric & class 1 PMU & class 0.1 PMU & $\mu$PMU \\
		\hline 
		Magnitude accuracy [\%] & 1\% & 0.1\% & 0.03\% \\
		Phase accuracy [rad] & $12 \cdot 10^{-3}$ & $1.5 \cdot 10^{-3}$ & $5.1 \cdot 10^{-4}$ \\
		\hline
	\end{tabular}
	\caption{PMU and $\mu$PMU 99th percentile accuracy \cite{sarri2016methods,von2017precision}. Magnitude accuracy is a reported as a percentage of the rated value.}
	\label{tab:pmu-error}
\end{table}

\section{Problem Statement}\label{section_problem}
Consider a power distribution network as described in \cref{section_model}, fulfilling Assumption \ref{ass:observability}. The identification problem amounts to reconstruct the admittance matrix from a sequence of voltage and current measurements corresponding to different steady states of the system \cite{yuan2016inverse, tomlin}.

Let $N$ be the number of samples, and let $v_\symt \in \bb{C}^n$ and $i_\symt \in \bb{C}^n$ be the vectors of current injections and voltages for $\symt = 1, \dots, N$. From \eqref{eq:current-voltage}, one can obtain
\begin{equation}
	\label{eq:current-voltage-samples}
	I = VY,
\end{equation}
where $V = [v_1, v_2, \dots, v_N]^\top \in \bb{C}^{N \times n}$, and $I = [i_1, i_2, \dots, i_N]^\top \in \bb{C}^{N \times n}$.

As described in \cref{subsec_data_collection}, the available current and voltage phasors are corrupted by measurement noise. Therefore, in place of the actual electrical variables $V$ and $I$, only noisy samples $\Tilde{V}$ and $\Tilde{I}$ are available, where
\begin{subequations}
\label{eq:noisy-and-actual-matrices}
\begin{align}
    \Tilde{V} &= V + \Delta V, \\
    \Tilde{I} &= I + \Delta I,
\end{align}
\end{subequations}
with $\Delta V \in \bb{C}^{N \times n}$ and $\Delta I \in \bb{C}^{N \times n}$ denoting the complex measurement noise. The network identification problem then translates into the estimation of $Y$ given $\Tilde V$ and $\Tilde I$.

\begin{assumption} \label{ass:constantness}
The admittance matrix is constant over the identification period.
\end{assumption}

In order to estimate $Y$ using the linear relationship \eqref{eq:current-voltage-samples}, it must be constant over both time and the obserbed variables. If not, the closest constant matrix will be identified instead, and the variations will add uncertainty to the estimate.

\section{Noise Model} 
\label{subsec_noise_carac}
By design, PMUs and $\mu$PMUs collect current and voltage measurements in polar coordinates, that is in terms of magnitude and phase \cite[Sec. 3.2]{sarri2016methods}. Previous studies have shown, with both theoretical and empirical arguments, that the measurement noise is approximately Gaussian in polar coordinates, with zero mean and constant variance \cite[Sec. 2.1]{sarri2016methods}. 

In \eqref{eq:current-voltage}, the admittance matrix $Y$ establishes a linear relationship between the real and the imaginary parts of $i$ and $v$, but the equation becomes non-linear if magnitude and phase are considered. To preserve linearity, one needs to transform the measurements and their associated noise from polar to Cartesian coordinates. This transformation changes the statistical distribution of the noise, which becomes non-Gaussian, with a non-zero mean, and with a time-varying, non-diagonal covariance matrix, as shown hereafter.

We consider a generic phasor measured by a $\mu$PMU. Without loss of generality, we will discuss only the case of a voltage phasor; the same arguments apply to currents. Let $\Tilde{v}$ and $\Tilde{\theta}$ denote the measured magnitude and phase, and let $v$ and $\theta$ be the noiseless (unobservable) variables. Then $\Tilde{v} = v + \epsilon$ and $\Tilde{\theta} = \theta + \delta$, where $\epsilon \sim \mathcal{N}(0, \sigma_\epsilon)$ and $\delta \sim \mathcal{N}(0, \sigma_\delta)$ are zero-mean Gaussian variables. Previous studies suggest that the following Assumption is usually satisfied \cite[Sec. 2.1]{sarri2016methods}.
\begin{assumption}\label{ass_independence_noise}
The samples taken at two different time steps on the same node, and the samples taken on two different nodes at the same time are independent.
\end{assumption}

Using Assumption \ref{ass_independence_noise}, we can focus on a single sample. Our aim is to write the measured phasor $\Tilde{v}e^{\imj \Tilde{\theta}}$ in Cartesian coordinate ($\Tilde{c} + \imj \Tilde{d}$). Let $v e^{j\theta} = c+jd$. Highlighting the noise component $\Delta c + \imj \Delta d$ gives $\Tilde{v}e^{\imj \Tilde{\theta}}= (c + \Delta c) + \imj (d + \Delta d)$, where
\begin{subequations}
\label{eq:noise-exact}
    \begin{align}
        \Delta c &= \Tilde{c} - c = v\cos\theta(\cos\delta - 1) - v\sin\theta\sin\delta + \epsilon\cos\theta\cos\delta - \epsilon\sin\theta\sin\delta, \\
        \Delta d &= \Tilde{d} - d = v\sin\theta(\cos\delta - 1) + v\cos\theta\sin\delta + \epsilon\sin\theta\cos\delta + \epsilon\cos\theta\sin\delta.
    \end{align}
\end{subequations}
From \eqref{eq:noise-exact}, it can be noted that $\Delta c$ and $\Delta d$ are not distributed as Gaussian variables, due to interaction terms like $\epsilon\sin\delta$. However, the first-order Taylor approximation of \eqref{eq:noise-exact} is a linear combination of Gaussian variables, suggesting that, for low noise levels, the distribution of $\Delta c$ and $\Delta d$ is closely approximated by a Gaussian variable. 
Similarly to \cite{wehenkel2020parameter}, to characterize the noise we compute the expected value of $\Delta c$ and $\Delta d$
\begin{subequations}
\label{eq:noise-exact-average}
    \begin{align}
        \E[\Delta c] &= v\cos\theta(e^{-\sigma^2_\delta/2}-1),\\
        \E[\Delta d] &= v\sin\theta(e^{-\sigma^2_\delta/2}-1),
    \end{align}
\end{subequations}
and the associated variance-covariance terms:
\begin{subequations}
\label{eq:noise-exact-variance}
    \begin{align}
        \var[\Delta c] &= v^2e^{-\sigma_\delta^2}[\cos^2\theta(\cosh\sigma_\delta^2-1)+\sin^2\theta\sinh{\sigma_\delta^2}] + \sigma_\epsilon^2e^{-\sigma_\delta^2}[\cos^2\theta\cosh\sigma_\delta^2 +\sin^2\theta\sinh{\sigma_\delta^2}], \\
        \var[\Delta d] &= v^2e^{-\sigma_\delta^2}[\sin^2\theta(\cosh\sigma_\delta^2-1)+\cos^2\theta\sinh{\sigma_\delta^2}] + \sigma_\epsilon^2e^{-\sigma_\delta^2}[\sin^2\theta\cosh\sigma_\delta^2+\cos^2\theta\sinh{\sigma_\delta^2}], \\
        \hspace{-6pt} \cov[\Delta c, \Delta d] &= \sin\theta\cos\theta e^{-2\sigma_\delta^2}[\sigma_\epsilon^2 + v^2(1-e^{\sigma_\delta^2})]. \label{eq:noise-exact-variance-cov}
    \end{align}
\end{subequations}
Unfortunately, the expressions in \cref{eq:noise-exact-average,eq:noise-exact-variance} are of no practical use, as they rely on the actual unobservable values $v$ and $\theta$. Similar issues arise in the context of state estimation based on Extended Kalman Filter (EKF): in particular, research on the filtering of radar signal can be adapted to our case \cite{julier2004unscented, longbin1998unbiased, duan2004comments, lerro1993tracking, sarri2016methods}. Following such developments, we compute the expectation of the average \eqref{eq:noise-exact-average} and the variance \eqref{eq:noise-exact-variance} conditioned on the measurements:
\begin{subequations}
\label{eq:average-true-bias}
    \begin{align}
        \E[\Delta c | \Tilde{v}, \Tilde{\theta}] &= \Tilde{v}\cos\Tilde{\theta}(e^{-\sigma^2_\delta}-e^{-\sigma^2_\delta/2}),\\
        \E[\Delta d | \Tilde{v}, \Tilde{\theta}] &= \Tilde{v}\sin\Tilde{\theta}(e^{-\sigma^2_\delta}-e^{-\sigma^2_\delta/2}).
    \end{align}
\end{subequations}
The same procedure can be applied to the variances:
\begin{subequations}
\label{eq:average-true-variance}
    \begin{align}
        \mt{Var}[\Delta c | \Tilde{v}, \Tilde{\theta}] &= \Tilde{v}^2e^{-2\sigma_\delta^2}[\cos^2\Tilde{\theta}(\cosh2\sigma_\delta^2-\cosh\sigma_\delta^2)
        +\sin^2\Tilde{\theta}(\sinh2\sigma_\delta^2-\sinh\sigma_\delta^2)] + \nonumber \\
        &+ \sigma_\epsilon^2e^{-2\sigma_\delta^2}[\cos^2\Tilde{\theta}(2\cosh2\sigma_\delta^2-\cosh\sigma_\delta^2)+\sin^2\Tilde{\theta}(2\sinh2\sigma_\delta^2-\sinh\sigma_\delta^2)], \\
        \mt{Var}[\Delta d | \Tilde{v}, \Tilde{\theta}] &= \Tilde{v}^2e^{-2\sigma_\delta^2}[\sin^2\Tilde{\theta}(\cosh2\sigma_\delta^2-\cosh\sigma_\delta^2)    +\cos^2\Tilde{\theta}(\sinh2\sigma_\delta^2-\sinh\sigma_\delta^2)] + \nonumber \\
        &+ \sigma_\epsilon^2e^{-2\sigma_\delta^2}[\sin^2\Tilde{\theta}(2\cosh2\sigma_\delta^2-\cosh\sigma_\delta^2) +\cos^2\Tilde{\theta}(2\sinh2\sigma_\delta^2-\sinh\sigma_\delta^2)], \\
        \mt{Cov}[\Delta c, \Delta d | \tilde{v}, \Tilde{\theta}] &= \sin\Tilde{\theta}\cos\Tilde{\theta} e^{-4\sigma_\delta^2}[\sigma_\epsilon^2 + (\Tilde{v}^2 + \sigma_\epsilon^2)(1-e^{\sigma_\delta^2})]. \label{eq:average-true-variance-cov}
    \end{align}
\end{subequations}

\cref{eq:average-true-bias} suggests that the measurement in Cartesian coordinates are biased, as the noise has a non-zero average. However, the arguments in Appendix \ref{appendix_bias} suggest that such bias is negligible for realistic noise levels. Moreover, it can always be computed and substracted from the data. Hence, in the following, the noise will be considered unbiased.

We finally model the noise on a phasor measurement as
\begin{equation} \label{eq:single-sample-covariance}
    \begin{bmatrix}
    \Delta c \\
    \Delta d
    \end{bmatrix} \sim \mathcal{N}\left(\bb{0}_2, \Sigma \right),
\end{equation}
with the elements of $\Sigma$ defined by \eqref{eq:average-true-variance}. The covariance matrix $\Sigma$ is not constant in time, but changes with the actual values of phase and magnitude: such property will be further discussed in \cref{section_methods}, while presenting the estimation methods. 

\section{Frequentist identification} \label{section_methods}
\subsection{Least squares} \label{subsec_methods}
From \eqref{eq:current-voltage-samples} and \eqref{eq:noisy-and-actual-matrices}, noisy data are related by the model
\begin{align}
    \label{eq:eiv_model}
    \tilde{I} - \Delta I = (\tilde{V} - \Delta V)Y.
\end{align}
For reconstructing $Y$, the works \cite{tomlin,online_rls} assume $\Delta V=0$ (i.e. $\tilde V = V$) and use Ordinary Least Squares (OLS) as well as its recursive and regularized variants. For standard OLS, one obtains
\begin{align}\label{def_ols}
    \hat Y_{\mt{OLS}} = \argmin_{\hat Y} \|\tilde I - \tilde V \hat Y\|_F^2.
\end{align}
However, if $\Delta V \neq 0$, the OLS introduces a bias. In this case, the Total Least Squares (TLS) is an unbiased estimator \cite{markovsky_tls_overview}: 
\begin{align}\label{def_tls}
    \hat Y_{\mt{TLS}} = \argmin_{\hat Y}\min_{I,V} \|[\tilde V - V,\tilde I-I]\|_F^2 \quad \mt{s.t.} \; I = V \hat Y.
\end{align}
Closed-form solutions for both estimators are well-known, and can be written column-wise, with centered data $\tilde V_c$ and $\tilde{I_c}$, as
\begin{subequations}
\begin{align}\label{def_tls_sol}
    \hat Y_{i,\mt{OLS}} &= (\tilde V_c^\top \tilde V_c)^{-1} \tilde V_c^\top \tilde I_{i,c}, \\
    \hat Y_{i,\mt{TLS}} &= (\tilde V_c^\top \tilde V_c - \sigma_{n+1}^2 \mc I)^{-1} \tilde V_c^\top \tilde I_{i,c},
\end{align}
\end{subequations}
where $\sigma_{n+1}$ is the smallest singular value of $[V_c, I_{i,c}]$. The error covariance, computed in terms of the exact values $V$ \cite{huffel_tls_cov}, is
\begin{subequations}\label{def_ols_tls_cov}
\begin{align}\label{def_ols_cov}
    \mt{Var}[\hat Y_{i,\mt{OLS}}] &\approx \frac{\sigma_{n+1}^2}{N}(V_c^\top V_c)^{-1}, \\
    \mt{Var}[\hat Y_{i,\mt{TLS}}] &\approx (1 + \|\hat Y_{i,TLS}\|^2)\frac{\sigma_{n+1}^2}{N}(V_c^\top V_c)^{-1}.\label{def_tls_cov}
\end{align}
\end{subequations}
It appears that $Y_{i,\mt{OLS}}$ has a smaller variance than $Y_{i,\mt{TLS}}$, but has a bias that grows with $\sigma_{n+1}$. In power systems, the data covariance $(V_c^\top V_c)^{-1}$ is very small compared to the noise variance (approximated by $\sigma_{n+1}$). Hence, the bias of the OLS can be very large, making the TLS a more suitable choice. However, the estimation of a sparse topology can be difficult, because the TLS estimate of zero elements can be very large due to the large variance. Possible solutions for this problem are regularization or more complex iterative methods such as the one presented in \cite{patopa}. Both methods assume that the samples are independent and identically distributed (i.i.d). However, as shown in section \ref{subsec_noise_carac}, µPMU measurements are not identically distributed.

\subsection{Maximum likelihood estimator} \label{subsec_mle_tls}

The high correlation between measurements observed in power systems suggest that large sample sizes $N$ may be needed in order to obtain a good estimate of $Y$. Moreover, the estimator should be unbiased and consistent \cite[chapter 7, 10]{statistical_inference_book}. The Maximum Likelihood Estimator (MLE) presents weak consistency conditions \cite{mle_consistency} that are satisfied for linear models. 
From \eqref{eq:current-voltage-samples} and \eqref{eq:noisy-and-actual-matrices}, noisy data are related by the model
\begin{align}
    \label{eq:eiv_model}
    \tilde{I} - \Delta I = (\tilde{V} - \Delta V)Y.
\end{align}
Considering Gaussian noise $\Delta V$ and $\Delta I$, as described in \cref{subsec_data_collection}, the MLE is a weighted Total Least Squares (TLS) estimator and thus shares most of its properties such as \emph{de}regularization \cite{markovsky_tls_overview} and unbiasedness \cite{crassidis_tls_cov}.

The variables $V$ and $I$ are deterministic, and the noises $\Delta V$ and $\Delta I$ are independent (see \cref{ass_independence_noise}). The likelihood $l(\tilde V, \tilde I | V,I,\hat Y)$ of $(\tilde V, \tilde I)$, can be written as follows.
\begin{align} \label{eq_abstract_likelihood}
l(\tilde V, \tilde I | V,I,\hat Y) \propto
\;& p(V + \Delta V|V,\hat Y)p(I + \Delta I|I,\hat Y), \\ 
& \mt{s.t.} \; (\tilde V - \Delta V)\hat Y = \tilde I - \Delta I. \nonumber
\end{align}

To work with real variables, we define
\begin{align}\label{def_aAb}
    \boldsymbol{v} &= \lt(\matb \Re\lt(\vect(V)\rt) \\ \Im\lt(\vect(V)\rt) \mate\rt), \quad \boldsymbol{i} = \lt(\matb \Re\lt(\vect(I)\rt) \\ \Im\lt(\vect(I)\rt) \mate\rt), \\
    \A &= \lt( \matb \Re(\mc I_n \otimes V) & - \Im(\mc I_n \otimes V) \\ \Im(\mc I_n \otimes V) & \Re(\mc I_n \otimes V) \mate \rt), \quad y = \lt(\matb \Re(\vect(Y)) \\ \Im(\vect(Y)) \mate \rt). \nonumber
\end{align}
The same transformations can be applied to $\tilde V$, $\tilde I$, $\Delta V$, $\Delta I$, and $\hat Y$, resulting in the vectors $\tilde{\boldsymbol{v}}$, $\tilde{\boldsymbol{i}}$, $\Delta \boldsymbol{v}$, $\Delta \boldsymbol{i}$, $\hat y$, as well as the matrices $\tilde \A$ and $\Delta \A$. Note that $\boldsymbol{v}$ and $\A$ contain the same elements but arranged differently. We will therefore use them interchangeably when describing optimization problems over $\Delta \boldsymbol{v}$ or $\Delta \A$. The matrix $\A$ is introduced to represent the product $VY$ with real and vectorized quantities as shown in \eqref{unreg_mle_notlog} below.

Using the vectorized notations \eqref{def_aAb}, we assume that $\Delta \boldsymbol{v} \sim \mc N(0,\Sigma_v)$ and $\Delta \boldsymbol{i} \sim \mc N(0,\Sigma_i)$, as per the approximate noise model discussed in \cref{subsec_data_collection}. The covariance matrices $\Sigma_v$ and $\Sigma_i$ are computed from \eqref{eq:average-true-variance} as explained in Appendix \ref{appendix:covariance}. The likelihood \eqref{eq_abstract_likelihood} becomes, up to a multiplicative constant, 


\begin{subequations}\label{unreg_mle_notlog}
\begin{align}
l(\tilde{\boldsymbol{v}}, \tilde{\boldsymbol{i}} | \boldsymbol{v},\boldsymbol{i},\hat y) = \; & e^{-\Delta \boldsymbol{i}^\top \Sigma_i^{-1} \Delta \boldsymbol{i}}e^{-\Delta \boldsymbol{v}^\top \Sigma_v^{-1} \Delta \boldsymbol{v}}, \\
&\mt{s.t.} \; \tilde{\boldsymbol{i}} - \Delta \boldsymbol{i} = (\tilde \A - \Delta \A) y.\label{unreg_mle_notlog_constraint}
\end{align}
\end{subequations}
The corresponding log-likelihood is 
\begin{align} \label{unreg_mle}
\mc L(\tilde{\boldsymbol{v}}, \tilde{\boldsymbol{i}} | \boldsymbol{v},\boldsymbol{i},\hat y) = & -\Delta \boldsymbol{i}^\top \Sigma_i^{-1} \Delta \boldsymbol{i} - \Delta \boldsymbol{v}^\top \Sigma_v^{-1} \Delta \boldsymbol{v},
\end{align}
subject to \eqref{unreg_mle_notlog_constraint}. For a fixed, albeit unknown $\boldsymbol{v}$ and $\boldsymbol{i}$, we use the shorthand notation $\mc L(\hat y, \Delta \boldsymbol{v}, \Delta \boldsymbol{i})$. Minimizing $-\mc L$ for $\Delta \boldsymbol{i}$, $\Delta \boldsymbol{v}$ and $\hat y$ yields the MLE.

\subsection{Error covariance analysis} \label{subsec_error_cov}
Let $\bar V = \frac{1}{N} \bb 1_N^\top V$ be the vector of voltage means, and $\Phi = \mc I_n \otimes (V - \bb 1_N \bar V)$ be the centered complex data matrix. The MLE with Gaussian noise as defined in \cref{subsec_noise_carac} is a weighted TLS estimator, and therefore, expressions for its error covariance can be found in the literature \cite{huffel_tls_cov, crassidis_tls_cov}. For a time-varying, not identically distributed case such as \eqref{unreg_mle}, \cite{crassidis_tls_cov} proves that the error covariance $\Sigma_{\mt{MLE}}$ is equal to the inverse of the Fischer information matrix $F_{\mt{MLE}}$, for which a first order approximation can be found in \cite[(48)]{crassidis_tls_cov}. We adapt this expression for the complex matrix $\Phi$ using the same transformation to real numbers as for $V$ in \eqref{def_aAb}, and obtain
\begin{align}\label{fischer_mle_unreg}
    F_{\mt{MLE}} = \sum_{q,\symi,\symt=1}^{2,n,N} \lt(\matb \frac{\Re(\Phi_{\symi N+\symt})^\top \Re(\Phi_{\symi N+\symt})}{\Re(z)^\top \mc R_{\Re,q\symi\symt} \Re(z)}  &
    \frac{\Re(\Phi_{\symi N+\symt})^\top \Im(\Phi_{\symi N+\symt})}{\Re(z)^\top \mc R_{\Re\Im,q\symi\symt} \Im(z)} \\ \star &
    \frac{\Im(\Phi_{\symi N+\symt})^\top \Im(\Phi_{\symi N+\symt})}{\Im(z)^\top \mc R_{\Im,q\symi\symt} \Im(z)} \mate\rt) \hspace{-3pt},
\end{align}
where $z = [\vect(Y)^\top, -1-\imj]^\top$ and $\mc R_{\Re,q\symi\symt}$, $\mc R_{\Re\Im,q\symi\symt}$, and $\mc R_{\Im,q\symi\symt}$ are diagonal and constructed from $\Sigma_v$ and $\Sigma_i$ as in \cite{crassidis_tls_cov}, and the $\star$ symbol means that $F_{\mt{MLE}}$ is symmetric. The exact expression is presented in Appendix \ref{appendix_cov_proof}.

One should note that $\hat F_{\mt{MLE}}$, computed from noisy data instead of the exact variables used for $F_{\mt{MLE}}$, is a good approximation only if the signal to noise ratio is high enough. Although this is typically not the case in distribution grids, $F_{\mt{MLE}}$ can still be used for theoretical purposes, such as design of experiments \cite{online_rls} for avoiding unobservablility problems (\cref{subsec_preprocess}). $F_{\mt{MLE}}$ also shows the following properties of the MLE.
\begin{lemma}\label{lem_covariance_shape}
The columns of $\hat Y_{\mt{MLE}}$ are independent. The variance of each column $(\hat Y_{\mt{MLE}}^\top)_\symi$ depends only on the same column $Y_{\cdot\symi}$ of the exact admittance matrix $Y$.
\end{lemma}
\begin{corollary}\label{cor_covariance_shape}
When $\mt{Cov}[\Re(\tilde V_\symt),\Im(\tilde V_\symt)] \approx 0$ for all $\symt$ (i.e. with small voltage angles), the variances of both real and imaginary parts of $(\hat Y_{\mt{MLE}}^\top)_\symi$ are monotone with their respective values, and their covariance is constant.
\end{corollary}
\begin{proof}
The proofs of both \cref{lem_covariance_shape} and \cref{cor_covariance_shape} can be found in Appendix \ref{appendix_cov_proof}.
\end{proof}
\begin{remark}
It follows from \eqref{eq:average-true-variance-cov} that the assumption of \cref{cor_covariance_shape} is also often true in in distribution networks, where voltage angles are very small.
\end{remark}

\cref{lem_covariance_shape} and \cref{cor_covariance_shape} are used in \cref{section_results} to identify and explain what factors can make the parameter identification very imprecise or even impossible.



\section{Bayesian estimation} \label{section_estimation}
Line admittances, even if measured, are typically known up to a tolerance. Some knowledge of $Y$'s structure, such as its sparsity, may also not be certain or precisely defined. This kind of uncertainty can be modeled via Bayesian prior distributions.


\subsection{Maximum \emph{a posteriori} estimation} \label{subsection_bayesian}
Following \cite{bayes_fang}, we describe how to compute Maximum \emph{A Posteriori} (MAP) estimates for the error-in-variables model \eqref{unreg_mle_notlog}. Using Bayes' rule, the posterior probability density is
\begin{align}\label{eq_abstract_posterior}
    p(V, I, \hat Y|\tilde V, \tilde I) = p(\tilde V, \tilde I | V,I,\hat Y)\frac{p(V,I)}{p(\tilde V, \tilde I)}p(\hat Y),
\end{align}
where we assume that the line parameters $y$ are independent of the grid state $(V,I)$ or its measurement $(\tilde V, \tilde I)$. The factor $\frac{p(V,I)}{p(\tilde V, \tilde I)}$ can be neglected as it is a quotient of non-informative priors \cite{bayes_fang}, defined as uniform distributions on the finite set of feasible voltages and currents. The negative log-posterior minimization of \eqref{eq_abstract_posterior} is then written as 
\begin{subequations}\label{bayes_optim_prob}
\begin{align}\label{bayes_optim_prob_cost}
    \min_{\hat y,\Delta \boldsymbol{v}, \Delta \boldsymbol{i}}  & -\mc L(\hat y, \Delta \boldsymbol{v}, \Delta \boldsymbol{i}) - \log \lt(p(\hat y)\rt),\\ \label{bayes_optim_prob_cons}
    & \tilde{\boldsymbol{i}} - \Delta \boldsymbol{i} = \lt( \tilde{\A} - \Delta \A \rt) \hat y,
\end{align}
\end{subequations}
with $\mc L$ defined by \eqref{unreg_mle}. Optimizing \eqref{bayes_optim_prob} provides a maximum \emph{a posteriori} (MAP) estimate $\hat y_{\mt{MAP}}$\footnote{The optimizer may not be unique, in this case $\hat y_{\mt{MAP}}$ is one of the elements of the set of optimizers.}.

The density function $p(\hat y)$ can take many forms. If it is Gaussian, then $-\log(p(\hat y))$ corresponds to a weighted ridge regularization \cite{bayesian_ridge}. In this paper, we will focus mainly on the element-wise Laplace distribution $p(\hat y_\symi) \propto e^{-\lambda |\hat y_\symi|}$, where all elements $p(\hat y_\symi)$ of the prior are assumed independent, and therefore $p(\hat y) = \prod_\symi p(\hat y_\symi)$ \cite{bayesian_lasso}. One obtains $-\log(p(\hat y)) = \lambda\|\hat y\|_1 + \mt{const}$, where the constant can be neglected in the optimization problem, obtaining an $\ell_1$ regularization term. The $\ell_1$ regularization can also be interpreted as a robustification of the MLE optimization problem $\min_{\hat y,\Delta \boldsymbol{v}, \Delta \boldsymbol{i}} \; -\mc L(\hat y, \Delta \boldsymbol{v}, \Delta \boldsymbol{i})$ \cite{bertsimas2018_robust_opt}.

Prior distributions are centered on the believed value of the exact admittance $y_i$, which can be different from zero, e.g., in case of an existing line. More generally, one can also believe that a linear combination of $y$ has a particular value \cite{bayesian_and_generalized_lasso}. For example, the belief of two lines having the same admittance is equivalent to that of their difference being zero, not knowing the actual value. The probability density $p(L\hat y - \mu)$ of a linear transformation $\hat y \rightarrow L \hat y - \mu$ can describe such a belief. The penalty function is then
\begin{align}\label{eq_ell1_proba}
    -\log(p(\hat y)) = \lambda\|L \hat y - \mu\|_1.
\end{align}

\begin{remark}
On one hand, if one has two priors about a line $\symi$, conditioned on independent events B and C, then similarly to \eqref{eq_abstract_posterior}, Bayes' rule gives that $p(\hat y_\symi)$ is proportional to $p(\hat y_\symi|B)p(\hat y_\symi|C)$. If $p(\hat y_\symi|B)$ and $p(\hat y_\symi|C)$ are both Laplace distributions, this is equivalent to adding a row to $L$ and $\mu$ in \eqref{eq_ell1_proba}. The same operation can be repeated for a larger number of priors on $\symi$.

On the other hand, if one has no prior about the line $\symi$, then the non-informative prior $p(\hat y_\symi)$ is a uniform distribution on the bounded support of $\hat y_\symi$, and can be factorized away in $p(\hat y)$, leaving $p(\hat y) \propto \prod_{\symk \neq \symi} p(\hat y_\symk)$.
\end{remark}

Introducing additional information may reduce the variance of the estimate. Using the approximation $\|\hat y\|_1 \approx \hat y^\top [|y|+\alpha]^{-1} \hat y$ similarly to \cite[Appendix A]{lasso_cov}, we can make $\log(p(\hat y))$ in \eqref{bayes_optim_prob} smooth. Hence, the Fisher information is given by the Hessian of the log-likelihood with smoothed prior.
\begin{align}
    F_{\mt{MAP}} \approx E\lt[\frac{\partial^2}{\partial \hat y^2}( -\mc L(\hat y, \Delta a, \Delta b) - \log (p(\hat y)))\rt].
\end{align}
Both the expected value and second derivative operators are distributive. For $p(\hat y) \propto e^{-\lambda \|L\hat y - \mu\|_1}$, $F_{\mt{MAP}}$ is therefore approximated as 
\begin{align}\label{fischer_mle_reg}
    F_{\mt{MAP}} \approx F_{\mt{MLE}} + \lambda L^\top ([|L y - \mu|] + \alpha \mc I])^{-1} L,
\end{align}
where $\alpha > 0$ is a small enough parameter. If the data matrix $\bb V$ is not full rank (so neither is $F_{\mt{MLE}}$), \eqref{fischer_mle_reg} can be used for assessing where prior knowledge is required to make $F_{\mt{MAP}}$ full rank.

\subsection{Prior distributions from known grid parameters}\label{subsec_hard_priors}
In this section, we will explain how to inject various forms of prior knowledge into the estimation problem. This knowledge can be in the form of partially known parameter values, other known grid properties such as sparsity, constraints on the signs of parameters, or on ratios between their values.

\emph{A priori} knowledge about the parameter $y_{\symi}$ can be done as follows.
\begin{align}\label{eq_prior_known_single}
    -\log(p(\hat y | \beta_{\symi})) = \lambda \|e_{\symi}^\top (\hat y - y_{\symi}) \|_1.
\end{align}
If a set $\mc H$ of parameters is known, and $\beta \in \R^{n^2}$ is a vector such that $\beta_{\symi} = y_{\symi}$ for all ${\symi} \in \mc H$ and zero otherwise, the log prior distribution becomes
\begin{align}\label{eq_prior_known}
    -\log(p(\hat y | \beta)) = \lambda \|[L_{\symi} e_{\symi}^\top]_{\symi \in \mc H}(\hat y - \beta) \|_1,
\end{align}
where each $L_{\symi} \in \R_+$ is a parameter modelling the confidence in the corresponding value $\beta_{\symi}$.

Sparsity is the belief that each parameter has a high probability to be zero. In the Bayesian framework, this translates into a zero-centered distribution, identical for all parameters, which is a standard $\ell_1$ penalty $\lambda \|\hat y \|_1$ in the log-space. If some lines are known to exist or not, the sparsity prior can be removed or strengthened only for them. Let $\beta_s \in \{0,1,K\}^{2n^2}$ be a sparsity pattern containing a $0$ for an existing line, an arbitrarily large $K$ for an absent line, and a $1$ otherwise. The corresponding sparsity promoting prior is
\begin{align}\label{eq_prior_sparse_onlysome}
    -\log(p(\hat y | \beta_s)) = \lambda \|[\beta_s]\hat y \|_1.
\end{align}


Prior knowledge may not directly concern the value of individual parameters. One could also know the sign of the elements of $y$. For instance, the conductance $\Re(Y_{\symi\symk})$ is always non-negative because of energy conservation laws, and in the absence of series capacitance, the susceptance $\Im(Y_{\symi\symk})$ is non-positive for all lines. In order to obtain a lower probability for the wrong signs, we will use the asymmetric Laplace distribution \cite{Kotz2001} with a very large parameter on the corresponding side. In this case, $\log(p(\hat y))$ becomes a sum of skewed absolute values written as
\begin{align}\label{eq_prior_signs}
    -\log(p(\hat y|s)) = \lambda \|\hat y\|_1 + K \sum_{\symi=1}^{2n^2}(1-s_\symi\mt{sgn}(\hat y_\symi)),
\end{align}
where $K$ is an arbitrarily large constant and $s$ is a vector with elements in $\{-1,1\}$ defining the believed signs of $y$.

Constraining the ratios between parameters is another form of prior knowledge that can be useful in many cases, for example if two lines are parallel, or if one knows the type of cable used for a certain line. For a ratio $\hat \rho_{\symi\symk} = \frac{y_{\symk}}{y_{\symi}}$, we introduce the prior
\begin{align}\label{eq_prior_single_ratio}
    -\log(p(\hat y|\hat \rho_{\symi\symk})) = \lambda \lt\|( \rho_{\symi\symk} e_{\symi} - e_{\symk})^\top \hat y \rt\|_1.
\end{align}
While the exact value of the resistance $R$ and the reactance $X$ of a line depends on its length, the type of cable gives their ratio $R$/$X$. A constant $R$/$X$ ratio for the entire network is a common assumption for distribution grids. 
By letting $\hat \rho$ be the estimated $R$/$X$ ratio, the corresponding prior is
\begin{align}\label{eq_prior_rx_ratio}
    -\log(p(\hat y|\hat \rho)) = \lambda \lt\|[ \hat \rho \mc I_{n^2} , - \mc I_{n^2} ] \hat y \rt\|_1.
\end{align}


\subsection{Data-driven prior distributions} \label{subsec_priors}
Bayesian priors can incorporate model beliefs resulting from other estimation methods. Thereafter, we show how MLE can be used to refine the priors described in \cref{subsec_hard_priors}. Sparsity promotion using \eqref{eq_prior_sparse_onlysome} can create a bias. This bias, can be reduced by using the weight $L = [|\hat y_{\mt{MLE}}|]^{-1}$, similarly to the adaptive Lasso method \cite{tomlin}. This yields
\begin{align}\label{eq_prior_adaptive}
    -\log(p(\hat y | y_{\mt{MLE}})) = \lambda\| [|\hat y_{\mt{MLE}}|]^{-1} \hat y\|_1
\end{align}
Note that this method is only asymptotically unbiased (as $N \rightarrow \infty$). For a finite sample size, it does not cancel the $\ell_1$ penalty's bias but still reduces it \cite{lasso_zou2006}.

\begin{remark}
For constructing the prior \eqref{eq_prior_adaptive}, $y_{\mt{MLE}}$ is considered constant rather than a random variable, preserving independence towards other potential priors.
\end{remark}

The hyperparameter $\lambda$ in \eqref{eq_ell1_proba} is very important because if it is too big, then the bias will be unnecessarily large, and if it is too small, some entries will not be effectively shrunk to zero. However, it can generally not be tuned using cross validation (which is the standard approach in Lasso) because $\Sigma_v$ and $\Sigma_i$ depend on the values of $\boldsymbol{v}$ and $\boldsymbol{i}$. For very sparse systems such as power grids, there may not exist a $\lambda$ such that the MAP estimate is both sparse enough and moderately biased. To solve this issue, one can add a prior on $\|\hat y\|_1$, centered on its believed value $\hat \gamma$. If the sign of every $y_i$ is known and with $s$ from \eqref{eq_prior_signs}, one can construct this prior as
\begin{align}\label{eq_prior_contrast}
    -\log(p(\hat y|\hat \gamma)) = \frac{\lambda'}{\hat \gamma} |s^\top \hat y - \hat \gamma|.
\end{align}
If $\lambda' > \lambda$, it limits the bias created by a too large $\lambda$, which makes its tuning much more tolerant to errors and allows for a more aggressive regularization.

Directly estimating $\hat \gamma$ with maximum likelihood as $\hat \gamma = \|\hat y_{\symi, \mt{MLE}}\|_1$ may be biased. The error $\hat y_{\mt{MLE}} - y$ is a Gaussian random variable centered at zero \cite{crassidis_tls_cov}. Hence, the terms of $\|\hat y_{\mt{MLE}}\|_1$ are absolute values of Gaussian variables centered at the corresponding $y_\symi$. For a random variable $X \sim \mc N (\bar X, \sigma_X^2)$, if $\bar X \gg \sigma_X$ then $|X|$ can be approximated as Gaussian, but if $\bar X = 0$, $|X|$ has a half-Gaussian distribution and $E[|X|] = \frac{\sqrt{2}}{\sqrt{\pi}}\sigma_X$. If $y$ is very sparse, many parameters are centered at zero, which creates a large bias on $\hat \gamma$. If one knows the sign $s_\symi \in \{-1,1\}$ that each parameter $y_\symi$ should have, one can replace $|\hat y_{\symi,\mt{MLE}}|$ by $s_h \hat y_{\symi,\mt{MLE}}$. The resulting estimate
\begin{align}\label{eq_gamma_estimate}
    \hat \gamma = s^\top \hat y_{\mt{MLE}}
\end{align}
has the same variance as $\|\hat y_{\mt{MLE}}\|_1$ but is unbiased because $E[\hat \gamma] = \sum_{\symi} s_\symi y_\symi = \|y\|_1$. Even without using \eqref{eq_prior_contrast}, this value can be useful to assess if the chosen $\lambda$ is too large. Note that the same argument applies if $\hat{y}_{\mt{MLE}}$ is split into groups of elements (e.g. into columns). However, if the groups are too small, the variance of each $\hat\gamma$ may be very large, leading to an erroneous prior.

\section{Structural priors} \label{subsec_structural}
Due to the structure of power networks, $Y$ has peculiar properties. If the network does not include phase-shifting transformers and power lines are not compensated by series capacitors, $Y$ is symmetric. Moreover, for networks where shunt elements are negligible, $Y$ is Laplacian \cite{kundur2007power}. Since phase-shifting transformers are usually employed in transmission systems, and shunt admittances are negligible for medium-sized grids, with line less than 60 km long, it is safe to assume that standard distribution networks have a Laplacian admittance matrix \cite{taleb2006performance}. 

\subsection{Formal definition}

Under the assumption that $Y$ is Laplacian, entries on and above the main diagonal of $Y$ can be derived from the elements below the diagonal. Therefore, in order to avoid redundant variables, one can proceed as in \cite{online_rls} and use duplication and transformation matrices $D$ and $T$ to remove the redundant entries from the identification problem and solve for $\hat y_r = [\Re(\ve(\hat Y))^\top, \Im(\ve(\hat Y))^\top]^\top$ instead. In case some entries of $Y$ are known to be zero, one can derive variants of $D$ and $T$ and also remove these zero entries from $\hat y$ by following a procedure similar to the one presented in \cite[Appendix 2]{online_rls}. In both cases, $\hat y = (\mc I_2 \otimes D \cdot T)\hat y_r$ and the equation \eqref{unreg_mle_notlog_constraint} becomes
\begin{align}\label{pf_eq_laplacian}
    \tilde{\boldsymbol{i}} - \Delta \boldsymbol{i} = (\tilde \A - \Delta \A) (\mc I_2 \otimes D \cdot T) \hat y,
\end{align}
The diagonal entries of $Y$ are often the largest, as they are the sum of all other entries on the same rows. According to \cref{cor_covariance_shape} and its following remark, the diagonal entries may not only have a large variance, but also cause one for all other elements. Using the $D$ and $T$ matrices also improves this last point.

\subsection{Implications for Bayesian priors}\label{subsec_priors_dt}

The MAP estimate can be computed using \eqref{bayes_optim_prob} while replacing \eqref{bayes_optim_prob_cons} by \eqref{pf_eq_laplacian}. One can also simply replace $\hat y$ by $\hat y_r$ in all the priors presented in \cref{section_estimation}. However, the Laplacianity of $Y$ opens additional opportunities. Similarly to \cref{subsec_hard_priors}, if the line susceptance is assumed negative, $\mt{diag}(Y_{\mt{MLE}})$ provides the following alternative estimate for $\hat \gamma$.
\begin{align} \nonumber
    \|\hat y\|_1 &= \sum_{\symi=1}^n\sum_{\symk=1,\symk \neq \symi}^n \lt\|\matb \Re(\hat Y_{\symi\symk}) \\ \Im(\hat Y_{\symi\symk}) \mate\rt\|_1 + \sum_{\symi=1}^n\lt\|\matb \Re(\hat Y_{\symi\symi}) \\ \Im(\hat Y_{\symi\symi}) \mate\rt\|_1, \\ 
    &= 2 \sum_{\symi=1}^n \lt\|[\Re(\hat Y_{\symi\symi}), \Im(\hat Y_{\symi\symi}) ]\rt\|_1.
\end{align}
Note that $\|\hat y\|_1 = 4\|\hat y_r\|_1$ because in $\hat y$, the diagonal elements make up half of the norm and the other elements are present twice. Furthermore, a prior on $\sum_{\symk=1, \symk\neq \symi}^n Y_{\symi\symk}$ for all $\symi \in \mc V$ can also limit the bias of a sparsity-promoting one. If the exact value is not available and with a Laplacian $Y$, this prior can be centered on $Y_{\mt{\symi\symi,MLE}}$. This yields
\begin{align}\label{eq_prior_nondiag}
    &-\log(p(\hat y_r | \mt{diag}(Y_{\mt{MLE}}))) = \\
    &\quad \lambda' \sum_{\symi=1}^n \lt|\frac{ \sum_{\symk=1, \symk\neq \symi}^n \Re(\hat Y_{\symi\symk})}{\Re(\hat Y_{\mt{MLE},\symi\symi})} - 1\rt| + \lt|\frac{ \sum_{\symk=1, \symk\neq \symi}^n \Im(\hat Y_{\symi\symk})}{\Im(\hat Y_{\mt{MLE},\symi\symi})} - 1\rt|. \nonumber
\end{align}
More details about the representation of $p(\hat y_r | \mt{diag}(\hat Y_{\mt{MLE}}))$ in the same form as \eqref{eq_ell1_proba} are presented in Appendix \ref{appendix_nondiag}.

In the simulation \cref{section_results}, we have a large network with a very sparse, Laplacian admittance matrix. Hence, we will us $D$ and $T$ with the prior $p(\hat y_r | s)$, as well as $p(\hat y_r | \hat y_{r,\mt{MLE}})$ and $p(\hat y_r | \mt{diag}(\hat Y_{\mt{MLE}}))$, which are built using the non-diagonal and diagonal elements respectively. Both priors are combined using the remarks in \cref{subsection_bayesian,subsec_priors}.

\section{Numerical methods} \label{section_numeric}

Before solving \eqref{bayes_optim_prob} to obtain an estimate, several improvements can be done by pre-processing available measurements. This section will first explain how centering and filtering the data, as well as removing hidden nodes can improve the estimation. Then, we will present and compare different algorithms to solve the optimization problem.

\subsection{Data pre-processing} \label{subsec_preprocess}

Some networks may have nodes with no load attached. If a node is unloaded, then the corresponding column of $V$ is a linear combination of the columns corresponding to neighboring nodes, as it is determined by a simple voltage divider. In this case, $\Phi$ in \eqref{fischer_mle_unreg} does not have full rank, so $F_{\mt{MLE}}$ is singular. In other words, the Cramer-Rao bound becomes infinite, at least for some parameters, which means that they cannot be reliably reconstructed by any unbiased estimator. In a similar spirit, if all the nodes are loaded but some nodes have much lower loads than others, $F_{\mt{MLE}}$ can be full rank, but some of its eigenvalues can be very small. When inverting $F_{\mt{MLE}}$, the small eigenvalues become very large, which means that the variance of the corresponding estimates will also be large.


From \eqref{fischer_mle_reg}, it is apparent that a prior on lines near an unloaded node may help with this issue as it potentially compensates for the rank deficiency of $F_{\mt{MLE}}$. However, this makes the prior the only source of information. Another solution is to take out these nodes from the problem, and identify a reduced matrix $Y_{\mt{red}}$ such that $I_{-\symi} = Y_{\mt{red}}V_{-\symi}$. In most applications, such as sensitivity analysis or control, this reduced matrix is sufficient as it only removes redundant parameters, while keeping an equivalent model for the remaining nodes.
$Y_{\mt{red}}$ can be computed using the Kron reduction method \cite{florian_kron}. 

For a power grid with a rated voltage $V_{\mt{rated}}$, the matrix $\tilde V^\top \tilde V$ is almost equal to $V_{\mt{rated}} \bb 1_n \bb 1_n^\top$. This matrix is then almost singular with one eigenvalue much larger than all others. This can be an issue for numerical stability both for minimizing \eqref{unreg_mle} or for solving \eqref{bayes_optim_prob}. If $Y$ is Laplacian, one can use $\tilde V_c = \tilde V - V_{\mt{rated}} \bb 1_N \bb 1_n^\top$ instead, as the second term is cancelled by $\bb 1_n^\top Y = \bb 0_n^\top$. The covariance $\Sigma_v$ still needs to be computed with $V$ and not $V_c$ due to the non-linear transformation \eqref{eq:average-true-variance}.


Finally, a low pass (moving average) filter can help reducing the measurement noise. Measurements from µPMUs are usually very frequent (50 to 120Hz) \cite{upmu_charac}, and load variation has an average period of a couple minutes. A low pass filter with a cutoff frequency equal to the one of the load profiles will not remove relevant information from the signal. However, as $K$ measurements are averaged, the noise variance of a filtered measurement is reduced since $[\sigma_{\varepsilon}^2, \sigma_\delta^2]_{\mt{filtered}} = K^{-1}[\sigma_{\varepsilon}^2, \sigma_\delta^2]$.

\subsection{Optimization algorithms} \label{subsec_algorithms}
If $-log(p(\hat y))$ is convex, the optimization problem \eqref{bayes_optim_prob} has a convex cost and bilinear constraints due to the multiplication of $\hat y$ and $\Delta \A$. Similarly to weighted TLS, no closed-form solution is available \cite{markovsky_tls_overview}.

The most basic algorithm for solving \eqref{bayes_optim_prob} is the alternate block coordinate descent (BCD), which sets $\Delta \A$ to constant to solve $\hat y$ for the next iteration ($\symkit$) and vice versa, as explained in \cite{S_TLS}. With $c = -\mc L - \log(p(\hat y))$, the update is
\begin{subequations} \label{bcd_update}
\begin{align}\label{bcd_update_a}
    \hspace{-4px} \Delta \A_{\symkit} &= \argmin_{\Delta \A} \; c \lt( \hat y_{\symkit-1}, \Delta \A, \tilde{\boldsymbol{i}} - \lt( \tilde{\A} \!-\! \Delta \A \rt)\hat y_{\symkit-1} \rt), \\
    \hat y_{\symkit} &= \argmin_{\hat y} \; c \lt(\hat y, \Delta \A_{\symkit}, \tilde{\boldsymbol{i}} - \lt( \tilde{\A} \!-\! \Delta \A_{\symkit} \rt)\hat y \rt). \label{bcd_update_y}
\end{align}
\end{subequations}
\eqref{bcd_update} shows two convex sub-problems that can be solved iteratively. However, the $\hat y$-subproblem \eqref{bcd_update_y} may not admit a closed-form solution, depending on $p(\hat y)$. When it does not, \eqref{bcd_update_y} can become the computational bottleneck.

To improve performance, one can use the approximation
$$\|L(\hat y - \mu)\|_1 \approx (\hat y - \mu)^\top L^\top [|L(\hat y_{\symkit-1}-\mu)| + \alpha \bb 1]^{-1} L(\hat y-\mu)$$
in the expression of $p(\hat y)$, with a small enough $\alpha$. This algorithm is called broken adaptive ridge regression (BAR)\cite{bar_asympto, bar_def}, and provides a closed-from approximate solution to \eqref{bcd_update_y}. If $L$ is diagonal, another possible alternative is use an operator splitting method such as ADMM \cite{admm_def, beck2017_prox_methods}. However, experimental evidence shows that more iterations are needed for ADMM to converge. A comparison of the three algorithms can be found in table \ref{table_algo_speeds}.

\begin{table}[H]
\centering
\begin{tabular}{|c||c|c|}
    \hline
    Algorithm & iterations to convergence & iterations/second \\
    \hline
    BCD & $\sim$10000 & 1.25 \\
    \hline
    BAR & $\sim$10000 & 30 \\
    \hline
    ADMM & $\sim$30000 & 28 \\
    \hline
\end{tabular}
\caption{Comparison of the speed of three algorithms on a 9 nodes network with $N = 400$ measurement samples, using a standard $\ell_1$ penalty with the same $\lambda$. In all cases, they are executed on a MacBook Pro with a 2.3GHz Intel i7 processor running Python 3.8. }
\label{table_algo_speeds}
\end{table}
\begin{remark} \label{rem_dual_lambda}
The convergence speed highly depends on the regularization parameter $\lambda$. To speed it up, a higher penalty can be applied in the first iterations, and then decreased to the optimal one. This optimal value can be computed using $\hat Y_{\mt{MLE}}$ as shown in \cref{subsec_priors}.
\end{remark}

If $L$ is diagonal, the proximal gradient method is also usable \cite{beck2017_prox_methods, proximal_methods}, but it requires to tune the step size on top of the regularization parameter. This can take many iterations, especially if the information contained in the data is very limited. Also, it relies on thresholding providing a closed-form proximal operator of the $\ell_1$ norm \cite{proximal_regularization}. If $L$ is not diagonal, the proximal operator becomes a piece-wise linear function with a number of pieces scaling with the number of combinations of signs in $L$. It therefore becomes quickly prohibitive to compute. In practice, non-diagonal priors are used for cancelling bias (\cref{subsec_priors}) or for keeping some parameters close to the same value, which can be needed for a three phased network if some sub-networks have three times the same line. The proximal Newton method \cite{proximal_methods} is not usable: it relies on the transformation of the proximal operator by the Hessian matrix of $c$, which is generally dense and therefore requires the analysis of up to $2^{2n^2}$ different sign combinations.

\section{Simulation results} \label{section_results}

\begin{figure}[hbt]
     \begin{center}
     \begin{subfigure}[b]{\textwidth}
     \begin{center}
        \includegraphics[width=0.5\textwidth]{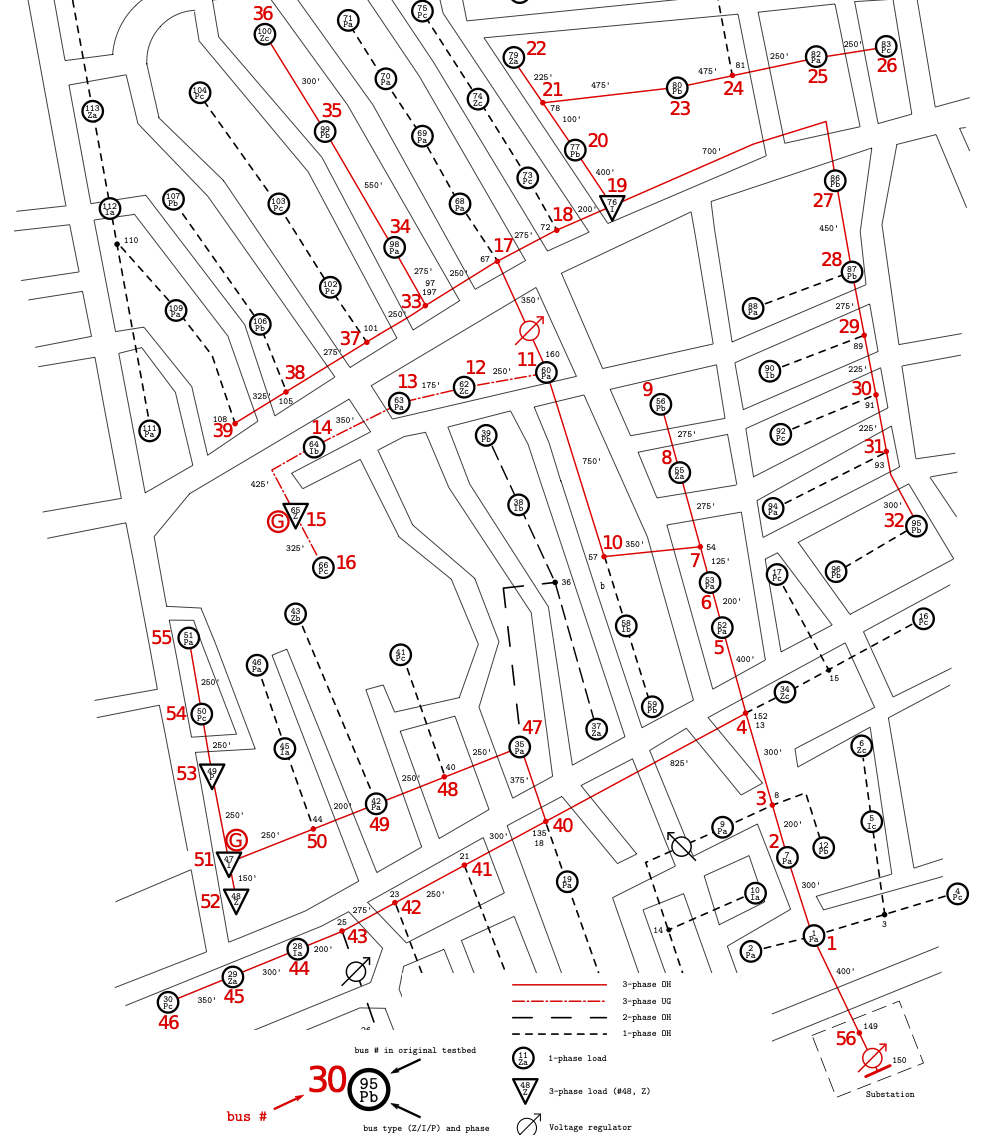}
        \caption{Map of the network from \cite{bolognani_network}.}
        \label{fig:ymat}
        \end{center}
     \end{subfigure}
     \newline
     \begin{subfigure}[b]{0.8\textwidth}
        \centering
        \input{loads.tex}
        \caption{Sample day of generated load profiles, with node 4's one highlighted.}
        \label{fig:loads}   
     \end{subfigure}
    \caption{Graphic representation of the simulation settings.}
    \label{fig:simsetup}
    \end{center}
\end{figure}

We apply the estimation method presented in \cref{section_estimation}, as well as approaches from other works, to a 56-node network with realistic parameters, admittances, noise levels and load profiles. In order to obtain the results, the BAR algorithm is implemented using a hardware-accelerated linear solver from NVidia's CUDA tool. The framework is programmed in Python and available on GitHub \cite{gitrepo}. It simulates the voltage and current measurements from the real network parameters, nominal loads, and given household load profiles. It then computes the various least squares and MLE and MAP estimates from the simulated measurements.

\subsection{Setup} \label{subsec_simsetup}

To simulate the identification problem, we use the IEEE 123 bus network. First, in \cref{subsec_results_balanced}, assuming that the loads on each phase are balanced, the three-phased part of the network is transformed into a 56 nodes, single-phase feeder using the method and parameters from \cite{bolognani_network}. Second, in \cref{subsec_results_unbalanced} the three-phased part is simulated with its original unbalanced loads. Finally, in \cref{subsec_results_reduced}, the identification is performed with partial information, which means that some nodes are unobserved and the equivalent admittance between the observed nodes is identified.

The load profiles for each node are generated with the GENETX generator \cite{loads_model}. It creates random realistic loads for households according to parameters like penetration of renewable energies (set to 35\%) or electric vehicles (40\%). The tool then generates a thousand one-minute-resolution demand profiles, for households situated in the Netherlands during week 12 of the year. Samples at different sampling frequencies are then extrapolated linearly from this data. To create the load profiles for each node (\cref{fig:loads}), the demand profiles of randomly selected households are summed until the nominal power is reached.

Voltage and current values are generated by simulating the network using the PandaPower library \cite{pandapower} with a measurement frequency of 50Hz. We then add 0.01\% of Gaussian noise in polar coordinates, as described in \cref{subsec_noise_carac}. 
Note that the noise generated by a $\mu$PMU depends on its rating. Assuming that the $\mu$PMUs are adapted to their nodes, we choose a rating of four times the nominal power. In order to reduce the computational complexity, the samples collected over a minute are averaged as proposed in \cref{subsec_preprocess} and the identification method is performed with the averaged samples.

A whole week of data with 50Hz sampling rate may seem excessive, but it is required to reach a practical 1\% to 2\% error. For different sample sizes, \cref{fig:sample_complexity} shows the expected relative Frobenius error $E[\varepsilon_F]$, where
\vspace{-0.5em}
\begingroup \belowdisplayskip=-2pt
\begin{align}\label{eq_frob_err_def}
    \varepsilon_F = \frac{\|\hat Y - Y\|_F}{\|Y\|_F}.
\end{align}
\endgroup
This quantity is obtained by sampling $\hat Y$ as a Gaussian random variable centered on $Y$ and with a covariance given by the Cramer-Rao lower bound $\mt{Cov}[y_{\mt{MLE}}] = F_{\mt{MLE}}^{-1}$, where $F_{\mt{MLE}}$ is computed using \eqref{fischer_mle_unreg}.

\setlength{\belowdisplayskip}{4pt}
\vspace{4em}
\begin{figure}[H]
    \centering
    \begin{tikzpicture}

\begin{axis}[
width=11cm,
height=5cm,
log basis x={10},
log basis y={10},
tick align=outside,
tick pos=left,
x grid style={white!69.0196078431373!black},
xmajorgrids,
xmin=1, xmax=60,
xmode=log,
xtick style={color=black},
xtick={1, 2, 3, 5, 7, 10, 15, 21, 30, 45, 60},
xticklabels={1, 2, 3, 5, 7, 10, 15, 21, 30, 45, 60},
y grid style={white!69.0196078431373!black},
ymajorgrids,
ymin=0.390317971942074, ymax=75,
ymode=log,
ytick style={color=black},
ytick={0.5,1,2,5,10,20,50},
yticklabels={0.5\%,1\%,2\%,5\%,10\%,20\%,50\%},
ylabel={relative estimation error},
xlabel={sample size [days]}
]
\path [draw=black, fill=black, opacity=0.2]
(axis cs:1,10.6243380732252)
--(axis cs:1,10.1906619267748)
--(axis cs:2,6.27658598606245)
--(axis cs:3,4.82492277670476)
--(axis cs:5,3.54984761610426)
--(axis cs:7,2.8947402671699)
--(axis cs:10,2.36591078535048)
--(axis cs:15,1.90234173993475)
--(axis cs:21,1.5536591287879)
--(axis cs:30,1.27623162426799)
--(axis cs:45,1.02757684962107)
--(axis cs:60,0.880454284223443)
--(axis cs:60,0.916045715776557)
--(axis cs:60,0.916045715776557)
--(axis cs:45,1.07792315037893)
--(axis cs:30,1.32726837573201)
--(axis cs:21,1.6218408712121)
--(axis cs:15,1.96015826006525)
--(axis cs:10,2.47808921464952)
--(axis cs:7,3.0167597328301)
--(axis cs:5,3.75015238389574)
--(axis cs:3,5.22157722329524)
--(axis cs:2,6.75041401393755)
--(axis cs:1,10.6243380732252)
--cycle;

\path [draw=blue, fill=blue, opacity=0.2]
(axis cs:1,75.1506044557625)
--(axis cs:1,72.0893955442375)
--(axis cs:2,44.3777030674326)
--(axis cs:3,34.1159049265212)
--(axis cs:5,25.0961805589574)
--(axis cs:7,20.459709400521)
--(axis cs:10,16.7310417207231)
--(axis cs:15,13.4535513917809)
--(axis cs:21,11.0110662486942)
--(axis cs:30,9.00309810163746)
--(axis cs:45,7.28322864177169)
--(axis cs:60,6.24513107382191)
--(axis cs:60,6.48436892617809)
--(axis cs:60,6.48436892617809)
--(axis cs:45,7.61977135822831)
--(axis cs:30,9.35340189836254)
--(axis cs:21,11.4789337513058)
--(axis cs:15,13.8864486082191)
--(axis cs:10,17.5339582792769)
--(axis cs:7,21.320290599479)
--(axis cs:5,26.5138194410426)
--(axis cs:3,36.9190950734788)
--(axis cs:2,47.7322969325674)
--(axis cs:1,75.1506044557625)
--cycle;

\addplot [semithick, black]
table {%
1 10.4075
2 6.5135
3 5.02325
5 3.65
7 2.95575
10 2.422
15 1.93125
21 1.58775
30 1.30175
45 1.05275
60 0.89825
};
\addlegendentry{50Hz}

\addplot [semithick, blue]
table {%
1 73.62
2 46.055
3 35.5175
5 25.805
7 20.89
10 17.1325
15 13.67
21 11.245
30 9.17825
45 7.4515
60 6.36475
};
\addlegendentry{1Hz}
\end{axis}

\end{tikzpicture}
    \caption{Graph in logarithm scale of the sample complexity for two sampling rates and a noise level of 0.01\%.}
    \label{fig:sample_complexity}
\end{figure}
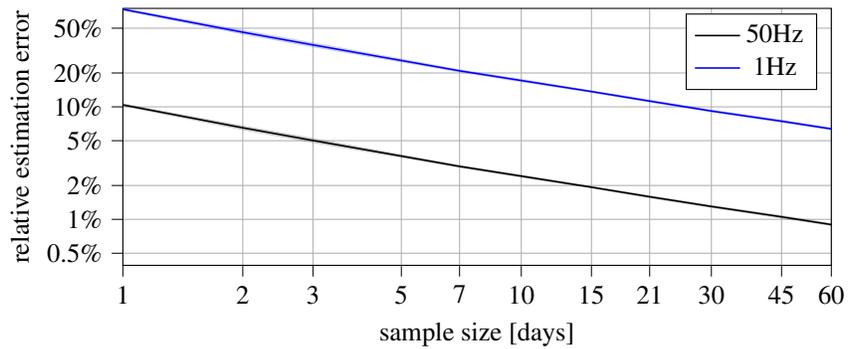

The MLE is efficient so its expected error should be equal to the Cramer-Rao bound \cite[chapter 7, 10]{statistical_inference_book}. However, this is true for the multivariate distribution. It then propagates non-linearly through \eqref{eq_frob_err_def} and the results in terms of norms may be slightly different. The Cramer-Rao bound provides a method-free indication of the results to expect, which shows the difficulty of the problem.

\subsection{Balanced network} \label{subsec_results_balanced}

With the high accuracy of µPMUs, and to the data preprocessing described in \cref{subsec_preprocess}, the Cramer-Rao lower bound is estimated around 3\% using \eqref{fischer_mle_unreg}. The MLE manages to retrieve an approximate but fair estimate with 5.77\% error (\cref{fig:tls_heatmap_1p}). Although the sparsity of the admittance matrix is above 98\%, this estimate is dense. 
\cref{fig:mle_heatmap_1p} shows the MAP estimate with a prior distribution as described in \cref{subsec_priors_dt}. It does not use any exact information, but only the MLE estimate described above as a starting point for the BAR algorithm, and achieves 1.21\% error, beating the Cramer-Rao bound.


\begin{figure}[htb]
     \centering
     \begin{subfigure}[t]{0.286\textwidth}
        \centering
        \adjincludegraphics[height=5cm,trim={0 0 {.25\width} 0},clip]{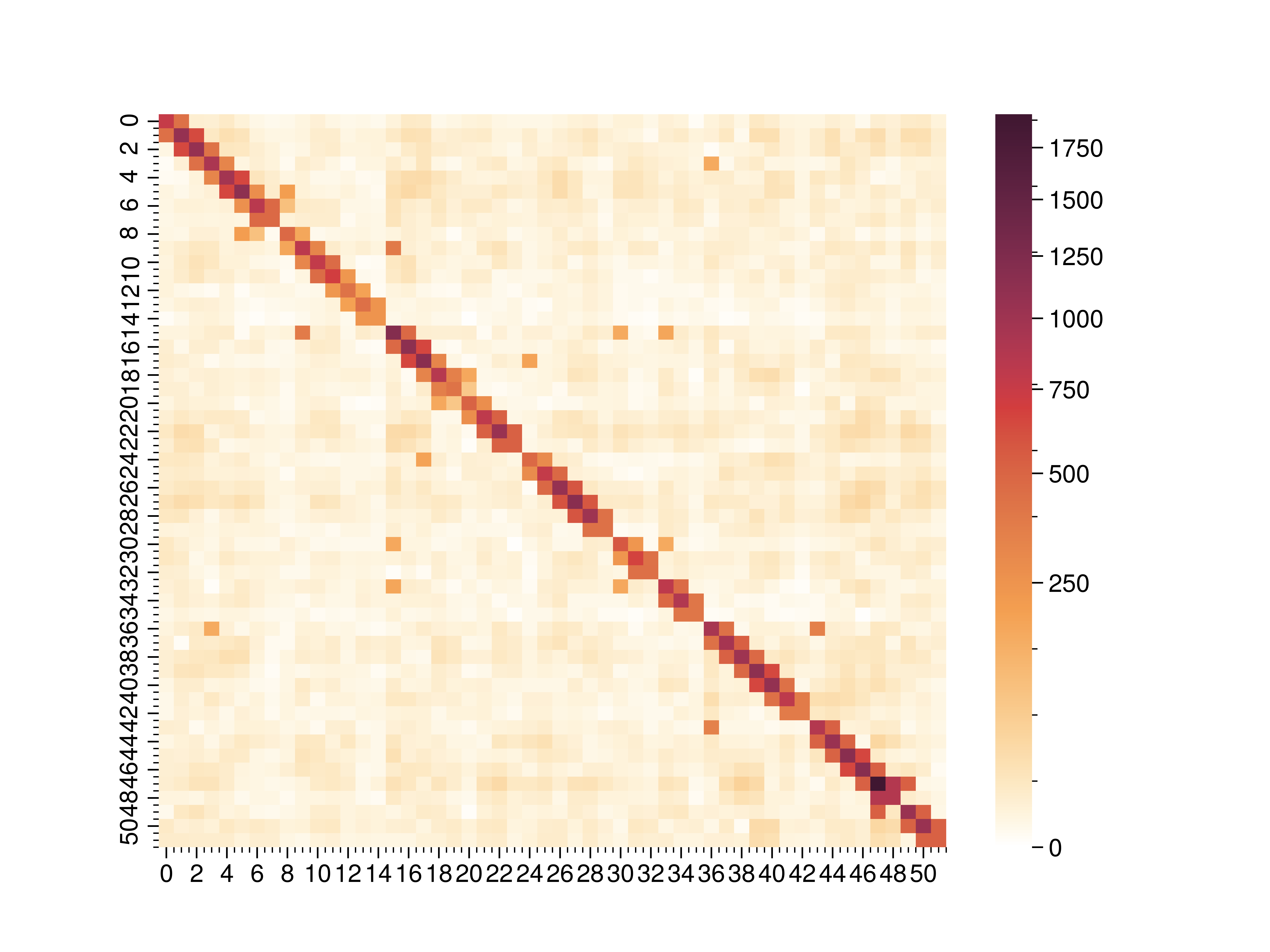}
        \caption{MLE estimate.}
        \label{fig:tls_heatmap_1p}   
     \end{subfigure}
     \begin{subfigure}[t]{0.286\textwidth}
        \centering
        \adjincludegraphics[height=5cm,trim={0 0 {.25\width} 0},clip]{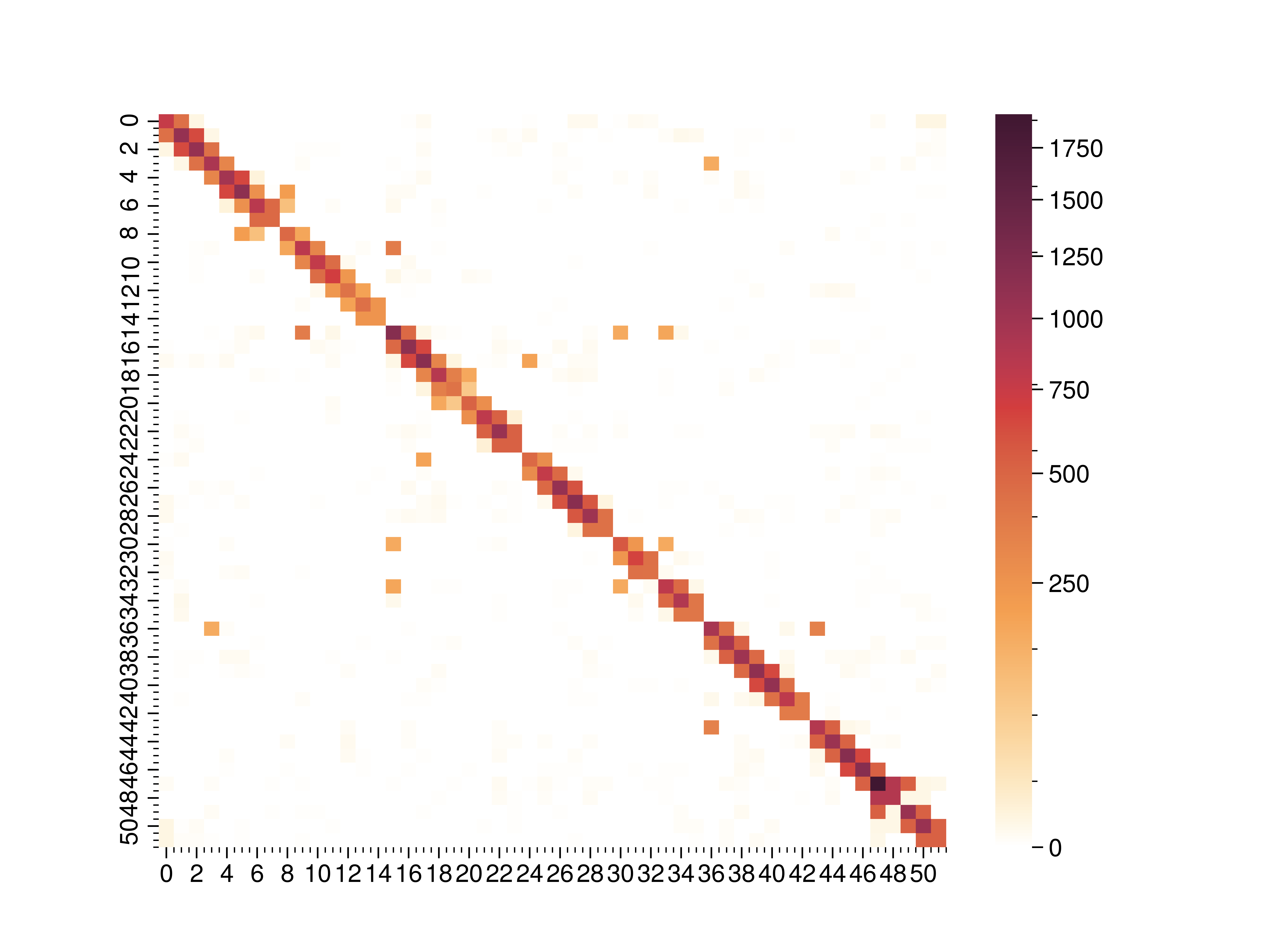}
        \caption{MAP estimate.}
        \label{fig:mle_heatmap_1p}   
     \end{subfigure}
     \begin{subfigure}[t]{0.34\textwidth}
        \centering
        \adjincludegraphics[height=5cm,trim={0 0 0 0},clip]{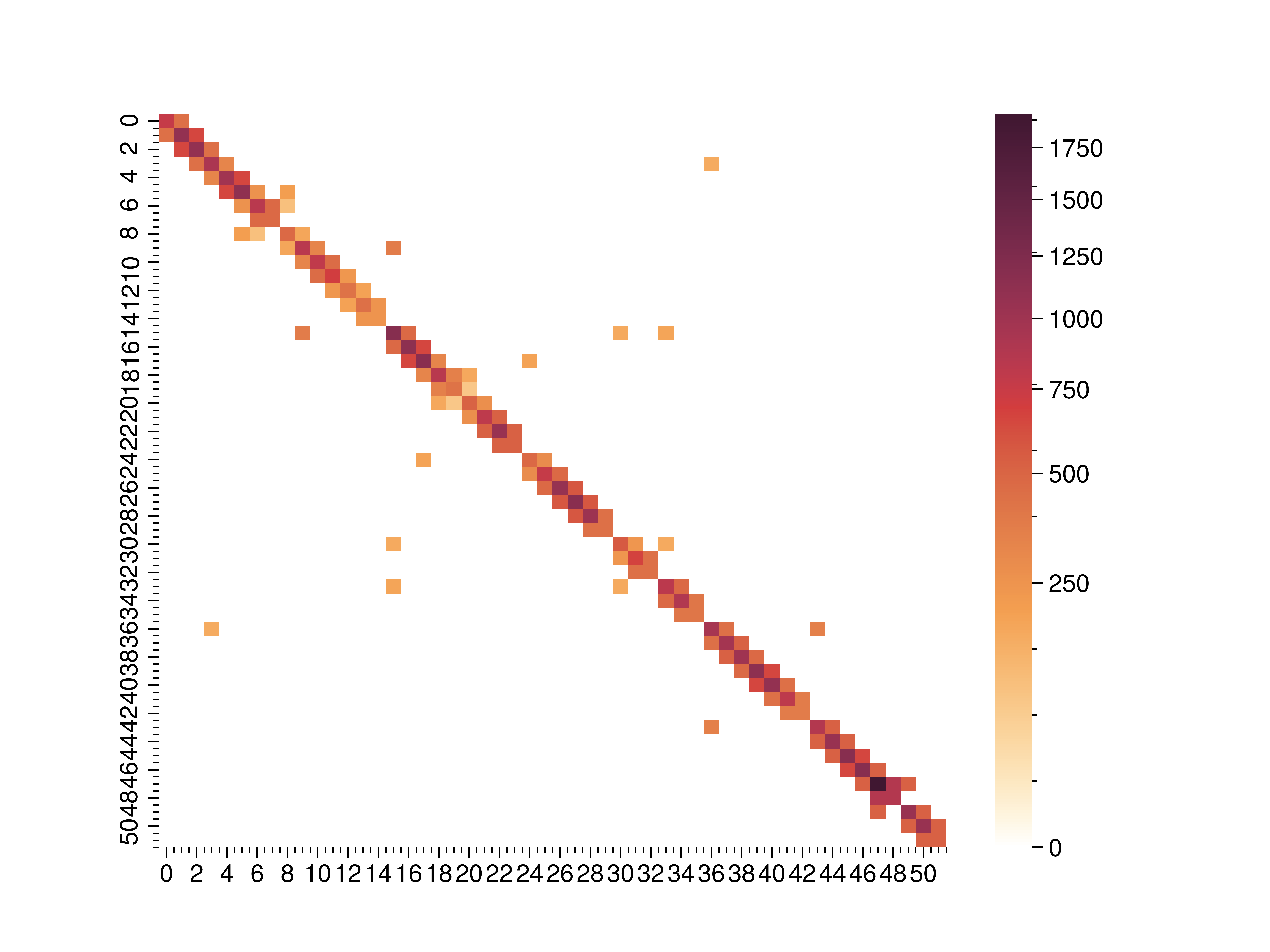}
        \caption{Real admittance matrix.}
        \label{fig:y_heatmap_1p}
     \end{subfigure}
    \caption{Heat-maps of real and estimated single-phase equivalent admittance matrices. Different colors correspond to log-spaced values.}
    \label{fig:heatmaps_1p}
\end{figure}

\subsection{Unbalanced network} \label{subsec_results_unbalanced}

Three-phased identification is a much more challenging problem. If some phases are not connected at any node or if any load is balanced, the voltage matrix $V$ will not have full rank. To circumvent this issue, we identify the closest balanced infrastructure network (i.e. a network with transposed lines) using sequence voltages, currents and admittances \cite{book_sequence}. In this case, zero, positive, and negative sequences can be estimated separately, which requires a rank of $V$ three times lower, and any unbalances in the network infrastructure will be considered as noise. The prior distribution presented in \cref{subsec_priors_dt} can be used for each sequence. \cref{fig:heatmaps_3p} shows the reconstructed phase admittance matrix. The error is 6.9\% for the MLE and 1.6\% for MAP.

\begin{figure}[htb]
     \centering
     \begin{subfigure}[t]{0.286\textwidth}
        \centering
        \adjincludegraphics[height=5cm,trim={0 0 {.25\width} 0},clip]{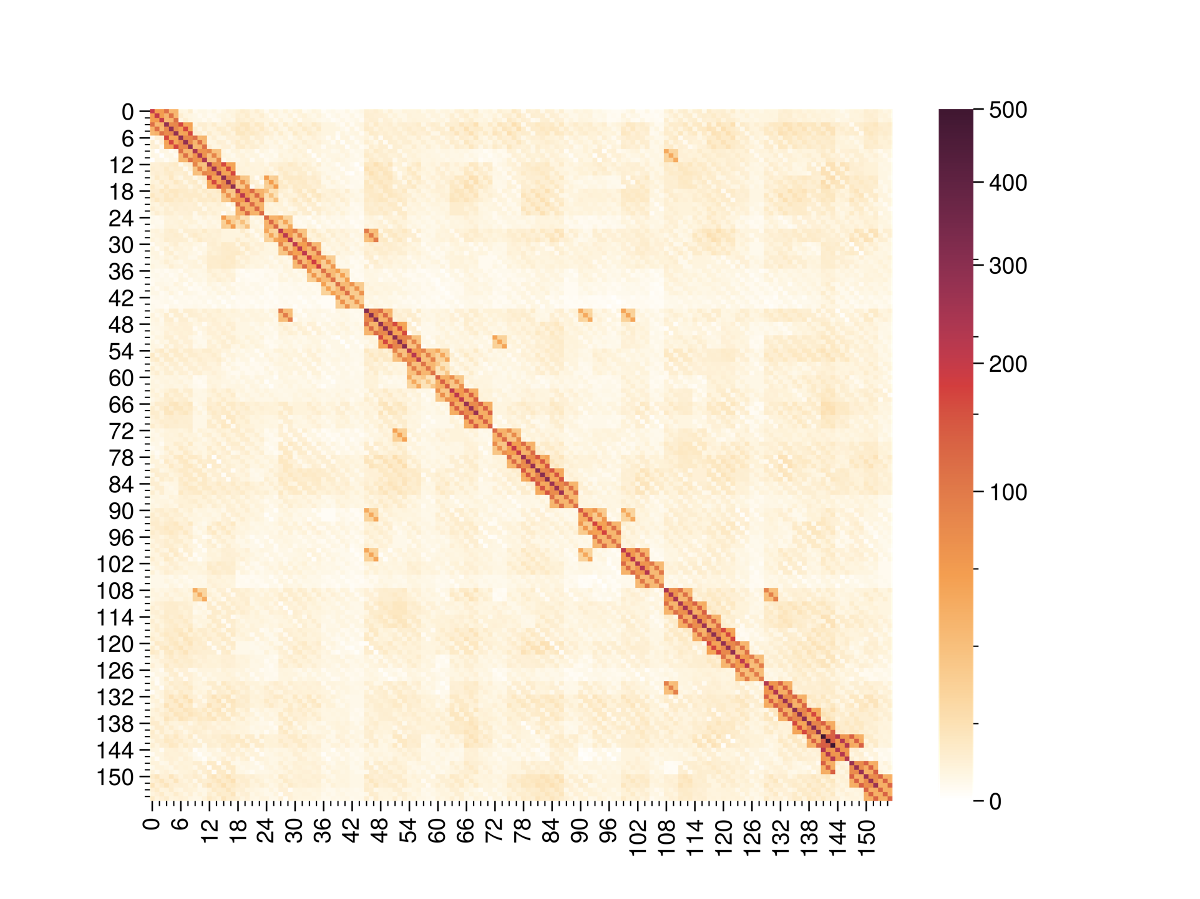}
        \caption{MLE estimate.}
        \label{fig:tls_heatmap_3p}   
     \end{subfigure}
     \begin{subfigure}[t]{0.286\textwidth}
        \centering
        \adjincludegraphics[height=5cm,trim={0 0 {.25\width} 0},clip]{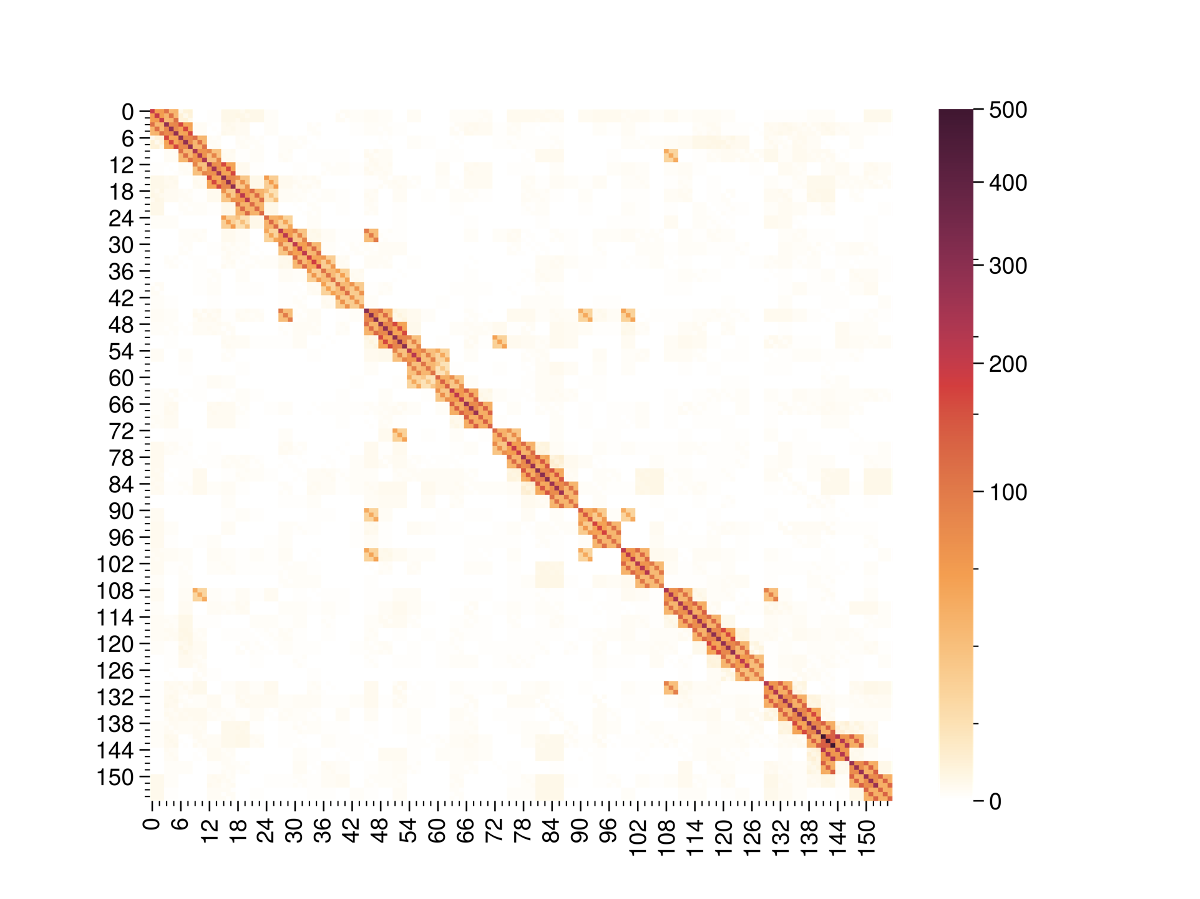}
        \caption{MAP estimate.}
        \label{fig:mle_heatmap_3p}   
     \end{subfigure}
     \begin{subfigure}[t]{0.34\textwidth}
        \centering
        \adjincludegraphics[height=5cm,trim={0 0 0 0},clip]{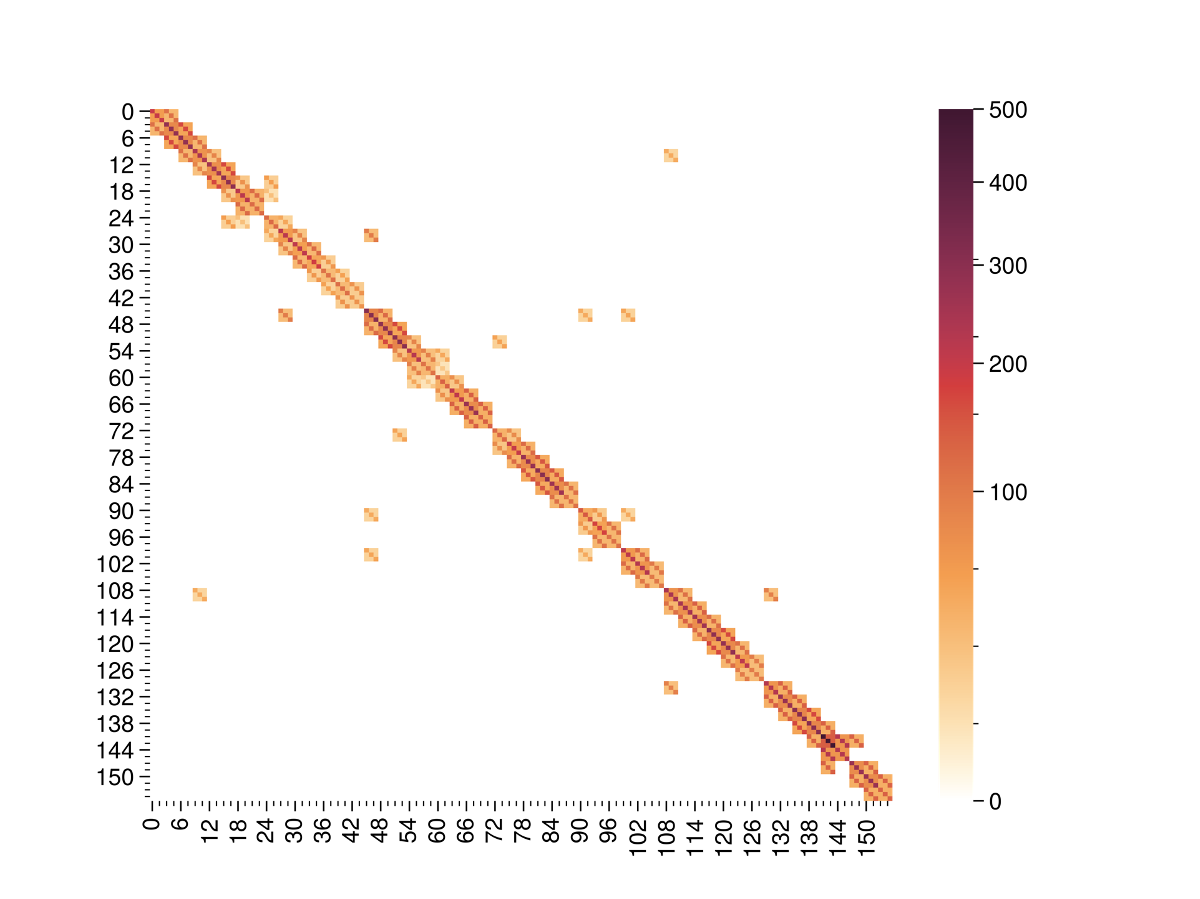}
        \caption{Real admittance matrix.}
        \label{fig:y_heatmap_3p}
     \end{subfigure}
    \caption{Heat-maps of real and estimated three-phased admittance matrices.}
    \label{fig:heatmaps_3p}
\end{figure}

One can observe that the error is similar to the single-phase cases, even though there are three times more parameters to estimate. This is due to the fact that the signal strength is similar but the admittance of the single phase equivalent is three times higher (three lines are considered as one). From \eqref{fischer_mle_unreg}, one can see that the information matrix is divided by 9 in the single-phase case, which increases the error. This also applies to parallel lines of the same phase.

\subsection{Reduced network} \label{subsec_results_reduced}

Following guidelines on optimal $\mu$PMU placement \cite{optimal_upmu_placement}, 40\% of the nodes\footnote{The nodes 1, 3, 4, 6, 8, 9, 10, 12, 15, 16, 17, 18, 19, 22, 24, 26, 28, 32, 36, 37, 39, 40, 43, 44, 46, 47, 49, 50, 51, 52, 53, and 55 are observed. (see \cref{fig:ymat}.)} in the network represented in \cref{fig:ymat} are not observed. The corresponding reduced admittance matrix (i.e. the matrix satisfying \eqref{eq:current-voltage} with reduced $i$ and $v$) is estimated with a 6.25\% error using MLE (\cref{fig:tls_heatmap_red}), and 2.49\% error using MAP (\cref{fig:mle_heatmap_red}). Note that because the network is now smaller and less sparse, the MLE is better but the contribution of the sparsity-promoting prior is less pronounced than for the full network.

\begin{figure}[htb]
     \centering
     \begin{subfigure}[t]{0.286\textwidth}
        \centering
        \adjincludegraphics[height=5cm,trim={0 0 {.25\width} 0},clip]{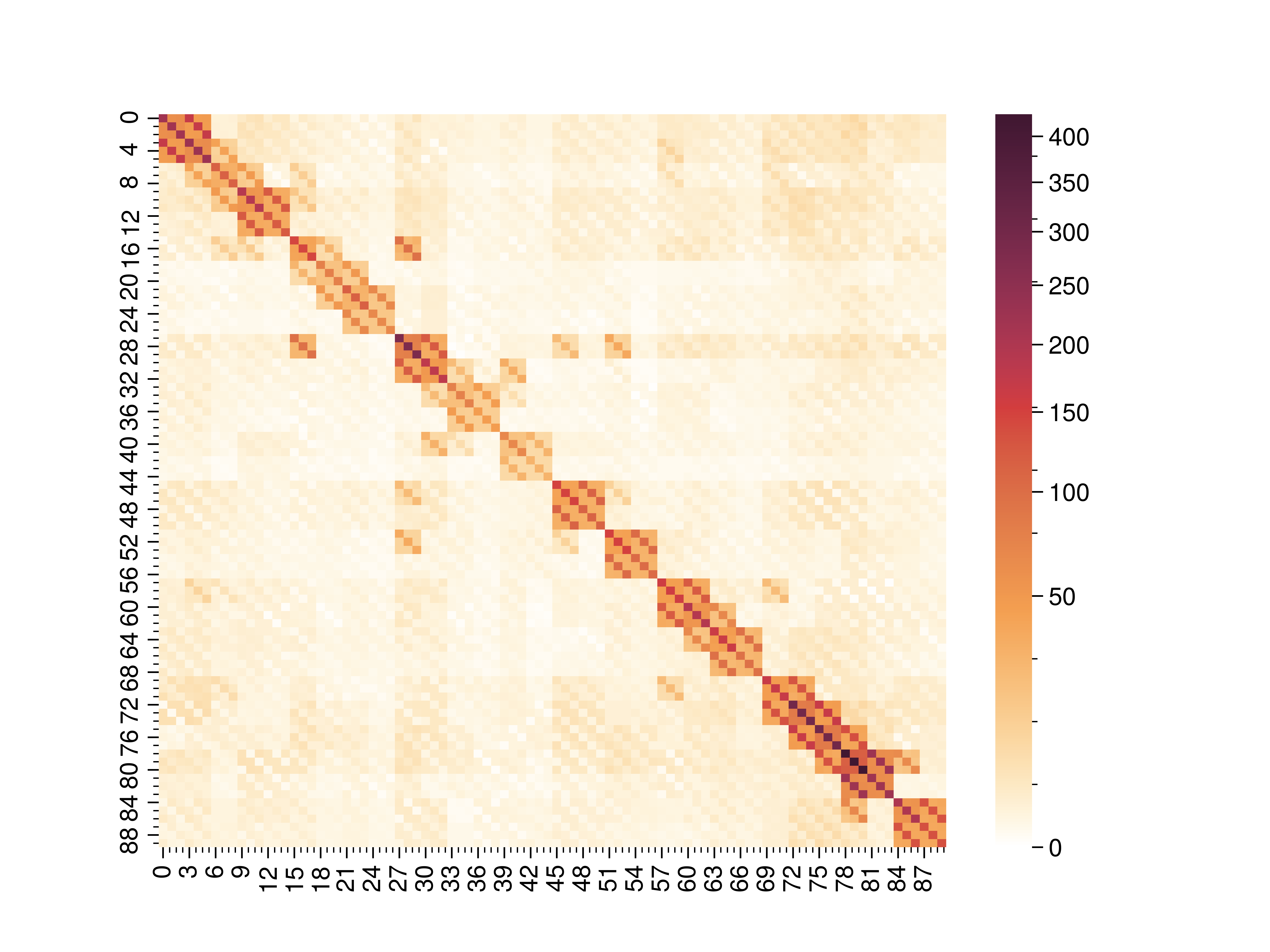}
        \caption{MLE estimate.}
        \label{fig:tls_heatmap_red}   
     \end{subfigure}
     \begin{subfigure}[t]{0.286\textwidth}
        \centering
        \adjincludegraphics[height=5cm,trim={0 0 {.25\width} 0},clip]{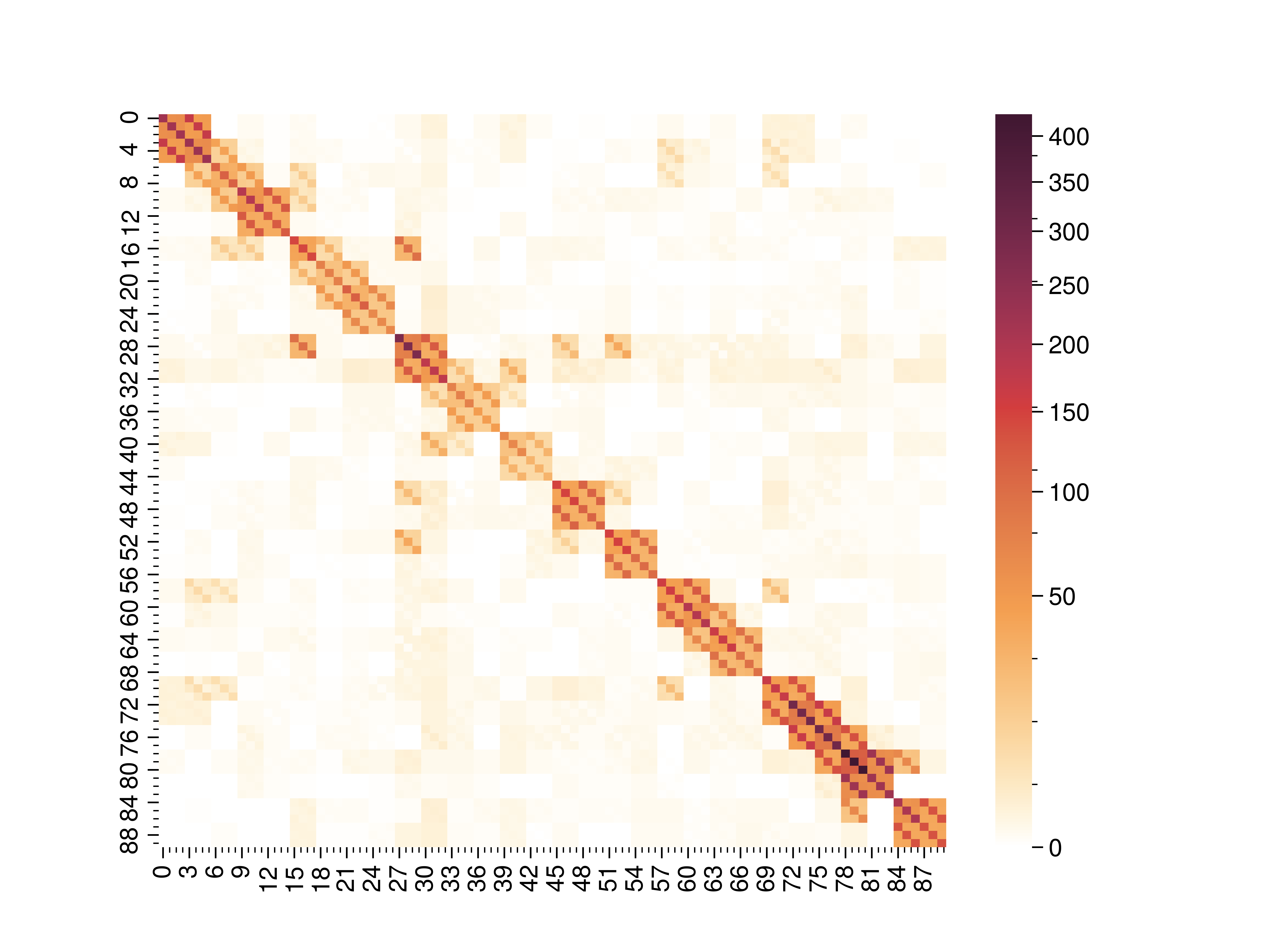}
        \caption{MAP estimate.}
        \label{fig:mle_heatmap_red}   
     \end{subfigure}
     \begin{subfigure}[t]{0.34\textwidth}
        \centering
        \adjincludegraphics[height=5cm,trim={0 0 0 0},clip]{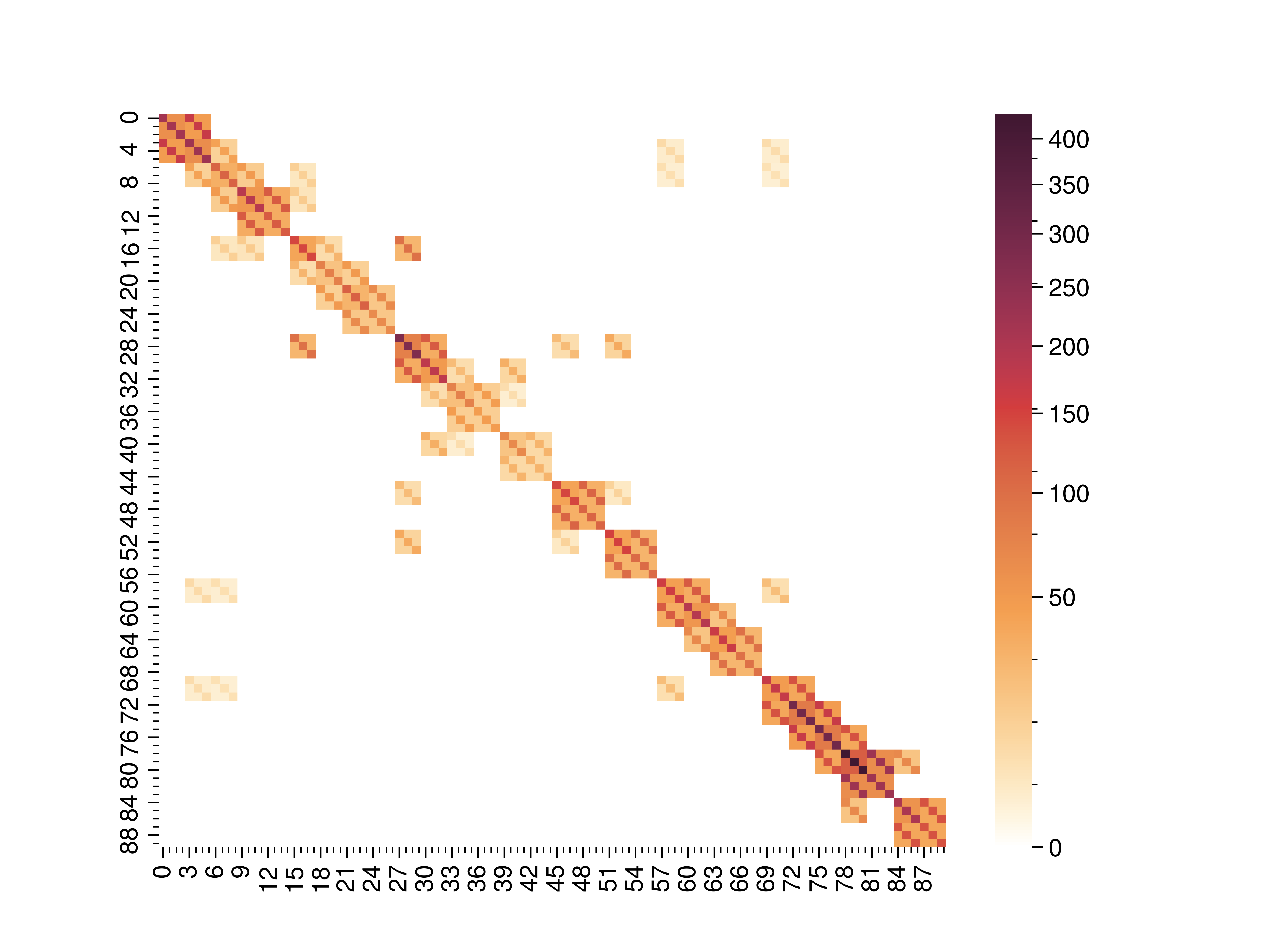}
        \caption{Real admittance matrix.}
        \label{fig:y_heatmap_red}
     \end{subfigure}
    \caption{Heat-maps of real and estimated reduced three-phased admittance matrices.}
    \label{fig:heatmaps_red}
\end{figure}


\section{Discussion} \label{section_discussion}

\subsection{Comparison with state-of-the-art methods} \label{subsec_comparison}

State-of-the-art approaches to network identification include, besides MLE (or its TLS approximation), other methods such as OLS or Lasso. \cref{fig:robustness} compares these three approaches with MAP using the sparsity-promoting prior described in \cref{section_results}. It shows the relative Frobenius error \eqref{eq_frob_err_def} for various noise level, as well as its standard deviation over 4 different simulations of the reduced network (\cref{subsec_results_reduced}). This figure highlights the low robustness to noise of non-EIV models, as well as the additional robustness provided by the regularization in MAP with physics-based priors.

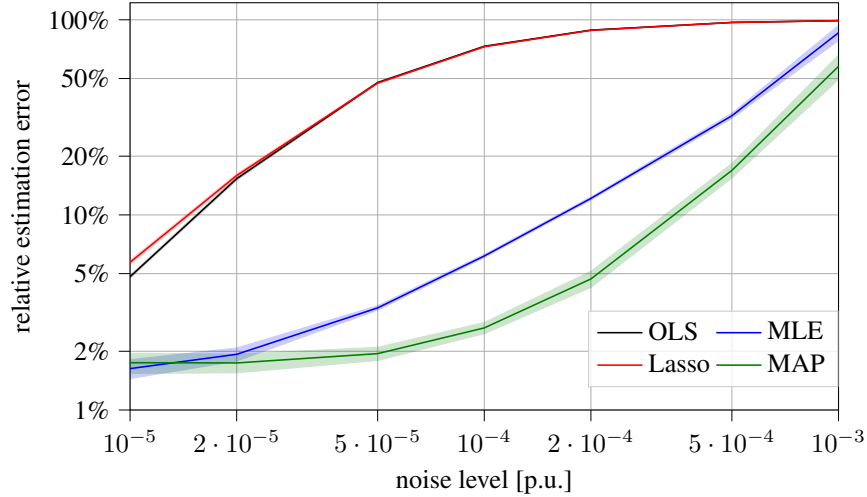
\begin{figure}[H]
    \centering
\begin{tikzpicture}

\begin{axis}[
legend cell align={left},
legend columns=2,
legend style={fill opacity=0.8, draw opacity=1, text opacity=1, at={(1,0.15)}, anchor=east, draw=white!80!black},
log basis x={10},
log basis y={10},
tick align=outside,
width=11cm,
height=7cm,
tick pos=left,
x grid style={white!69.0196078431373!black},
xmajorgrids,
xmin=1e-05, xmax=0.001,
xmode=log,
xtick style={color=black},
y grid style={white!69.0196078431373!black},
ymajorgrids,
ymin=0.01, ymax=1.22304407455456,
xtick={0.00001,0.00002,0.00005,0.0001,0.0002,0.0005,0.001},
xticklabels={$10^{-5}$,$2\cdot 10^{-5}$, $5\cdot 10^{-5}$, $10^{-4}$, $2\cdot 10^{-4}$, $5\cdot 10^{-4}$, $10^{-3}$},
ymode=log,
ytick style={color=black},
ytick={0.01,0.02,0.05,0.1,0.2,0.5,1},
yticklabels={1\%,2\%,5\%,10\%,20\%,50\%,100\%},
xlabel={noise level [p.u.]},
ylabel={relative estimation error}]
]
\path [draw=black, fill=black, opacity=0.2]
(axis cs:1e-05,0.0489842193570679)
--(axis cs:1e-05,0.0474157806429321)
--(axis cs:2e-05,0.150257670780787)
--(axis cs:5e-05,0.471762255228189)
--(axis cs:0.0001,0.72848141447824)
--(axis cs:0.0002,0.884079015518638)
--(axis cs:0.0005,0.969093798079768)
--(axis cs:0.001,0.989841886116992)
--(axis cs:0.001,0.990158113883008)
--(axis cs:0.001,0.990158113883008)
--(axis cs:0.0005,0.969906201920232)
--(axis cs:0.0002,0.886720984481362)
--(axis cs:0.0001,0.734418585521759)
--(axis cs:5e-05,0.481337744771811)
--(axis cs:2e-05,0.156442329219213)
--(axis cs:1e-05,0.0489842193570679)
--cycle;

\path [draw=blue, fill=blue, opacity=0.2]
(axis cs:1e-05,0.0181138357147217)
--(axis cs:1e-05,0.0144861642852783)
--(axis cs:2e-05,0.0178133931252681)
--(axis cs:5e-05,0.0324907630408086)
--(axis cs:0.0001,0.0604252717712433)
--(axis cs:0.0002,0.119290537823628)
--(axis cs:0.0005,0.312522590272842)
--(axis cs:0.001,0.781395318220831)
--(axis cs:0.001,0.935454681779169)
--(axis cs:0.001,0.935454681779169)
--(axis cs:0.0005,0.331977409727158)
--(axis cs:0.0002,0.123659462176372)
--(axis cs:0.0001,0.0627247282287567)
--(axis cs:5e-05,0.0342592369591914)
--(axis cs:2e-05,0.0207866068747318)
--(axis cs:1e-05,0.0181138357147217)
--cycle;

\path [draw=red, fill=red, opacity=0.2]
(axis cs:1e-05,0.0585284677683639)
--(axis cs:1e-05,0.0558215322316361)
--(axis cs:2e-05,0.1560751000016)
--(axis cs:5e-05,0.469397341177589)
--(axis cs:0.0001,0.724730151519017)
--(axis cs:0.0002,0.881558992781175)
--(axis cs:0.0005,0.968255033785326)
--(axis cs:0.001,0.989496464289287)
--(axis cs:0.001,0.989853535710714)
--(axis cs:0.001,0.989853535710714)
--(axis cs:0.0005,0.968994966214674)
--(axis cs:0.0002,0.884191007218825)
--(axis cs:0.0001,0.730669848480983)
--(axis cs:5e-05,0.478802658822411)
--(axis cs:2e-05,0.1623248999984)
--(axis cs:1e-05,0.0585284677683639)
--cycle;

\path [draw=green!50.1960784313725!black, fill=green!50.1960784313725!black, opacity=0.2]
(axis cs:1e-05,0.0195454769981818)
--(axis cs:1e-05,0.0154045230018182)
--(axis cs:2e-05,0.0155048763060678)
--(axis cs:5e-05,0.0179358996070273)
--(axis cs:0.0001,0.0246114661349286)
--(axis cs:0.0002,0.0424619603388562)
--(axis cs:0.0005,0.154629095035332)
--(axis cs:0.001,0.493550858982598)
--(axis cs:0.001,0.660899141017402)
--(axis cs:0.001,0.660899141017402)
--(axis cs:0.0005,0.183470904964668)
--(axis cs:0.0002,0.0514380396611438)
--(axis cs:0.0001,0.0280885338650714)
--(axis cs:5e-05,0.0209641003929727)
--(axis cs:2e-05,0.0193451236939322)
--(axis cs:1e-05,0.0195454769981818)
--cycle;

\addplot [semithick, black]
table {%
1e-05 0.0482
2e-05 0.15335
5e-05 0.47655
0.0001 0.73145
0.0002 0.8854
0.0005 0.9695
0.001 0.99
};
\addlegendentry{OLS}
\addplot [semithick, blue]
table {%
1e-05 0.0163
2e-05 0.0193
5e-05 0.033375
0.0001 0.061575
0.0002 0.121475
0.0005 0.32225
0.001 0.858425
};
\addlegendentry{MLE}
\addplot [semithick, red]
table {%
1e-05 0.057175
2e-05 0.1592
5e-05 0.4741
0.0001 0.7277
0.0002 0.882875
0.0005 0.968625
0.001 0.989675
};
\addlegendentry{Lasso}
\addplot [semithick, green!50.1960784313725!black]
table {%
1e-05 0.017475
2e-05 0.017425
5e-05 0.01945
0.0001 0.02635
0.0002 0.04695
0.0005 0.16905
0.001 0.577225
};
\addlegendentry{MAP}
\end{axis}

\end{tikzpicture}
    \caption{Comparison of existing methods for various noise levels and 7 days of data.}
    \label{fig:robustness}
\end{figure}

\begin{remark}
Many datasets suffer from incomplete measurements (e.g. data collected using smart meters miss the voltage phase). The EIV model generates estimates of the exact value of both currents and voltages, which could replace large variance pseudo-measurements. However, doing so may greatly reduce the precision of the estimate and require very large sample sizes.
\end{remark}

\subsection{Prior improvement by measuring selected line admittances} \label{subsec_measure}
If one needs to further improve the estimate and the sample size is limited, the remaining solution is to collect further information on topology or parameters of the electric network and integrate the results into the Bayesian prior (\cref{subsec_hard_priors}). 
To avoid collecting too much data, the additional information can be focused on estimates with the largest error. Although the error covariance is practically not computable (\cref{subsec_error_cov}), its properties show that in theory one expects the estimation error to be concentrated on (i) high admittance elements such as short lines or switches, (ii) lines connecting a node with a load much lower than their power flow, and (iii) lines with a very low admittance, which may be strongly affected by regularization and estimated as zero.

The first point follows from \cref{cor_covariance_shape}, which implies that a high admittance element also affects the estimation error of the other lines connected to the same node. Point (ii) is equivalent to the remark in \cref{subsec_preprocess}, that if a node has a low load, the error variances of the admittances of all lines connected to it may be very high.

According to these guidelines, we incorporate the knowledge of the lines connected to nodes 1, 50 and 51. Indeed, these lines are short and have a very high power flow due to neighbouring large loads or external grid connection, hence falling in the categories (i) and (ii). In total, we add information about 5 of the 870 possible connections, among which 35 actually exist, and obtain 2.24\% error, which means a 10\% improvement compared to the value without prior knowledge. As a comparison, we add the measurements of 5 random other lines and obtain 2.4\% error, which means that the estimate improves by only 3.6\%. The guidelines (i), (ii), and (iii) can therefore help choosing lines to measure to achieve a better estimate.



\section{Conclusions} \label{section_conclu}
The penetration of distributed generation and smart devices in the distribution grid calls for the introduction of advanced control schemes, which require the exact topology and line parameters. Such information is often unavailable for distribution networks: as direct measurement is infeasible, data-driven estimators are needed. 

In this work, we proposed to exploit samples collected by micro-PMUs. Considering a realistic statistical model for the noise affecting both current and voltage measurements, we built maximum-likelihood and Bayesian estimators. We argued that the latter can outperform the former, due to their ability to exploit features of the grid, such as sparsity, as well as available information on specific lines. Our argument is substantiated by numerical simulations on benchmark grids: even without any network-specific prior information, Bayesian methods outperformed state-of-the-art estimators with realistic noise levels.

Further research may focus on effective methods to selectively collect live measurements for improving the quality of estimates, as well as the development of alternative noise models for different sensors and physical quantities such as power. Formally defining the network's observability could lift the limitation to fully observed grids, introducing pseudo-measurements and Bayesian priors on the grid's state. Learning a reduced network connecting only specific nodes could also provide an answer to missing measurements.


\appendix

\section{Covariance matrix}\label{appendix:covariance}
In order to solve the Maximum Likelihood problem \eqref{unreg_mle_notlog} and all its subsequent refinements, one needs the covariance matrices $\Sigma_v$ and $\Sigma_i$. The construction of the two is identical, thus we will focus on $\Sigma_v = \cov[\Delta \boldsymbol{v}] \in \bb{R}^{2nN \times 2nN}$ only. 

From \cref{subsec_noise_carac}, the $\Sigma_v$ is sparse, having non-zero elements only on three diagonals:
\begin{align}\label{mle_error_covariance}
    \Sigma_v = \begin{pmatrix}
  \matb \mt{Var}[\Re(\tilde V_{11})] &  \multicolumn{2}{c}{\text{\kern0.5em\smash{\raisebox{-1ex}{\Large $\bb 0$}}}} \\ & \ddots & \\  \multicolumn{2}{c}{\text{\kern0.5em\smash{\raisebox{1ex}{\Large $\bb 0$}}}} & \mt{Var}[\Re(\tilde V_{Nn})] \mate 
  & \rvline
  & \matb \mt{Cov}[\Re(\tilde V_{11}),\Im(\tilde V_{11})] &  \multicolumn{2}{c}{\text{\kern0.5em\smash{\raisebox{-1ex}{\Large $\bb 0$}}}} \\ & \ddots& \\  \multicolumn{2}{c}{\text{\kern0.5em\smash{\raisebox{1ex}{\Large $\bb 0$}}}} & \mt{Cov}[\Re(\tilde V_{Nn}),\Im(\tilde V_{Nn})] \mate \\
  \hline
  \star
  &  \rvline
  & \matb \mt{Var}[\Im(\tilde V_{11})] &  \multicolumn{2}{c}{\text{\kern0.5em\smash{\raisebox{-1ex}{\Large $\bb 0$}}}} \\ & \ddots & \\  \multicolumn{2}{c}{\text{\kern0.5em\smash{\raisebox{1ex}{\Large $\bb 0$}}}} &  \mt{Var}[\Im(\tilde V_{Nn})] \mate 
\end{pmatrix},
\end{align}
where $\star$ denotes a symmetric element. The main diagonal of $\Sigma_v$ hosts the variance while the $nN$th super- and sub-diagonals, provide the covariance between the real and the imaginary part of the measurements. Such particular structure makes it possible to split $\Sigma_v$ into four diagonal blocks and makes it easy to find $\Sigma_v^{-1}$ analytically. 

It is also interesting to note that, up to a permutation of the elements in $\Delta \boldsymbol{v}$, $\cov[\Delta \boldsymbol{v}]$ can be written as a block diagonal matrix where the 2-by-2 blocks are given by \eqref{eq:average-true-variance-cov}.

\section{Proof in \cref{subsec_error_cov}}\label{appendix_cov_proof}
\subsection{Proof of \cref{lem_covariance_shape}}

Since $\Phi = \mc I_n \otimes (V - \bb 1_N \bar V)$, one has that $\Phi_{\symi N+\symt}^\top \Phi_{\symi N+\symt}$ is block diagonal for all $\symi,\symt$, with $n$ blocks of size $N$ by $n$. Hence, $F_{\mt{MLE}}$ is a 2-by-2 block matrix, with each block being a block diagonal matrix with $n$ blocks, and its inverse has the same shape. The error covariance matrix $\Sigma_y$ is then written using the compact notation $\hat Y_\symi^\Re + \imj \hat Y_\symi^\Im = (\hat Y^\top)_{\symi,\mt{MLE}}$ as
\begin{align}\label{mle_error_covariance}
    \Sigma_y = \begin{pmatrix}
  \matb \mt{Var}[\hat Y_1^\Re] &  \multicolumn{2}{c}{\text{\kern0.5em\smash{\raisebox{-1ex}{\Large $\bb 0$}}}} \\ & \ddots & \\  \multicolumn{2}{c}{\text{\kern0.5em\smash{\raisebox{1ex}{\Large $\bb 0$}}}} & \mt{Var}[\hat Y_n^\Re] \mate 
  & \rvline
  & \matb \mt{Cov}[\hat Y_1^\Re,\hat Y_1^\Im] &  \multicolumn{2}{c}{\text{\kern0.5em\smash{\raisebox{-1ex}{\Large $\bb 0$}}}} \\ & \ddots& \\  \multicolumn{2}{c}{\text{\kern0.5em\smash{\raisebox{1ex}{\Large $\bb 0$}}}} & \mt{Cov}[\hat Y_n^\Re,\hat Y_n^\Im] \mate \\
  \hline
  \star
  &  \rvline
  & \matb \mt{Var}[\hat Y_1^\Im] &  \multicolumn{2}{c}{\text{\kern0.5em\smash{\raisebox{-1ex}{\Large $\bb 0$}}}} \\ & \ddots & \\  \multicolumn{2}{c}{\text{\kern0.5em\smash{\raisebox{1ex}{\Large $\bb 0$}}}} &  \mt{Var}[\hat Y_n^\Im] \mate 
\end{pmatrix} .
\end{align}

The zeros introduced by the Kronecker product for constructing $\Phi$ are not random variables (i.e. their variance and covariance is zero). Together with assumption \ref{ass_independence_noise}, this yields
\begin{align}\label{fischer_mle_covariance}
    \mc R_{\Re,1\symi\symt} = \begin{pmatrix}
  \matb [ \bb 0_{n(\symi-1)} ] &\hspace{-6px} \\ &\hspace{-6px} \mt{Var}[ \Re(\tilde V_\symt) ] \mate 
  & \hspace{-10px} \text{\Large $\bb 0$} \\
  \text{\Large $\bb 0$} & \hspace{-10px} 
  \matb [ \bb 0_{n(n-\symi)} ] &\hspace{-6px} \\ &\hspace{-6px} \mt{Var}[\Re( \tilde I_{\symi\symt})] \mate
\end{pmatrix},
\end{align}
where $\mt{Var}[\Re(\tilde{V}_\symt)]$ is diagonal from \cref{ass_independence_noise}. Thus, $\mc R_{\Re,1\symi\symt}$ is also diagonal. $\mc R_{\Re,2\symi\symt}$ has the same expression as \eqref{fischer_mle_covariance}, but with $\Im( \tilde I_{\symi\symt})$ replacing $\Re( \tilde I_{\symi\symt})$. Expressions of the same form can be derived using $\mt{Cov}[\Re(\tilde V_\symt), \Im(\tilde V_\symt)]$ and $\mt{Var}[ \Im(\tilde V_\symt)]$ for $\mc R_{\Re\Im,q\symi\symt}$ and $\mc R_{\Im,q\symi\symt}$, respectively. With $q \in \{1,2\}$ and $\mc Q_{q\symi\symt}$ such that $\mc Q_{1\symi\symt} = \mt{Var}[\Re( \tilde I_{\symi\symt})]$ and $\mc Q_{2\symi\symt} = \mt{Var}[\Im( \tilde I_{\symi\symt})]$, we define
\begin{subequations}\label{fisher_D_def}
\begin{align}
    \mc D_{\Re,q\symi\symt}(Y_{\cdot\symi}) &= \Re(Y_{\cdot\symi})^\top \mt{Var}[ \Re(\tilde V_\symt) ] \Re(Y_{\cdot\symi}) + \mc Q_{q\symi\symt}, \\
    \mc D_{\Re\Im,q\symi\symt}(Y_{\cdot\symi}) &= \Re(Y_{\cdot\symi})^\top \mt{Cov}[ \Re(\tilde V_\symt),\Im(\tilde V_\symt) ] \Im(Y_{\cdot\symi}) + \mc Q_{q\symi\symt}, \\
    \mc D_{\Im,q\symi\symt}(Y_{\cdot\symi}) &= \Im(Y_{\cdot\symi})^\top \mt{Var}[ \Im(\tilde V_\symt) ] \Im(Y_{\cdot\symi}) + \mc Q_{q\symi\symt},
\end{align}
\end{subequations}
Excluding the zeros in \eqref{fischer_mle_covariance} yields $D_{\Re,1\symi\symt}(Y_\symi) = \Re(z)^\top \mc R_{\Re,1\symi\symt} \Re(z)$, and similarly for $D_{Q,q\symi\symt}(Y_\symi)$ with any $Q \in \{\Re,\Re\Im,\Im\}$ and $q \in \{1,2\}$.

Moreover, $\Phi_{\symi N+\symt}^\top \Phi_{\symi N+\symt}$ also has the same sparsity pattern as the $\mc R$ matrices. Therefore, for any $\symk$, the $\symk^{th}$ $n$-by-$n$ block of $F_{\mt{MLE}}$ is only nonzero for the terms of the sum \eqref{fischer_mle_unreg} where $\symi = \symk$. This means that
\begin{align}\label{mle_error_subcovariance}
    \mt{Var}\!\lt[\matb \hat Y_\symi^\Re \\ \hat Y_\symi^\Im \mate\rt] \!=\! \lt(\!\sum_{q,\symt=1}^{2,N}
    \matb \frac{\Re(\Phi_{\symi N+\symt})^\top \Re(\Phi_{\symi N+\symt})}{D_{\Re,q\symi\symt}(Y_{\cdot\symi})} &\hspace{-5px}
    \frac{\Re(\Phi_{\symi N+\symt})^\top \Im(\Phi_{\symi N+\symt})}{D_{\Re\Im,q\symi\symt}(Y_{\cdot\symi})} \\ \star &\hspace{-5px}
    \frac{\Im(\Phi_{\symi N+\symt})^\top \Im(\Phi_{\symi N+\symt})}{D_{\Im,q\symi\symt}(Y_{\cdot\symi})} \mate\rt)^{\hspace{-5px}-1} \hspace{-8px}.
\end{align}
Equations \eqref{mle_error_covariance} and \eqref{mle_error_subcovariance} respectively show not only that the columns $(\hat{Y}^\top)_\symi$ of $\hat Y_{\mt{MLE}}$ are statistically independent, but also that their variance does not depend on the exact values of one another, which finishes the proof.



\subsection{Proof of \cref{cor_covariance_shape}}
Let $Y_{\cdot\symi}^1$ and $Y_{\cdot\symi}^2$ be such that $|\Re(Y_{\cdot\symi}^2)| \geq |\Re(Y_{\cdot\symi}^1)|$ and $|\Im(Y_{\cdot\symi}^2)| \geq |\Im(Y_{\cdot\symi}^1)|$ element-wise. From \eqref{fisher_D_def} and assuming $\mt{Cov}[\Re(V_\symt),\Im(V_\symt)] = 0$, it follows that, for all $q$ and $\symt$,
\begin{subequations} \label{fisher_D_ineq}
\begin{align}
    \mc D_{\Re,q\symi\symt}(Y_{\cdot\symi}^2) &\geq  D_{\Re,pit}(Y_{\cdot\symi}^1), \\
    \mc D_{\Re\Im,q\symi\symt}(Y_{\cdot\symi}^2) &= \mc D_{\Re\Im,q\symi\symt}(Y_{\cdot\symi}^1), \\
    \mc D_{\Im,q\symi\symt}(Y_{\cdot\symi}^2) &\geq \mc D_{\Im,q\symi\symt}(Y_{\cdot\symi}^1),
\end{align}
\end{subequations}
because $\mt{Var}[ \Re(\tilde V_\symt) ]$ and $\mt{Var}[ \Im(\tilde V_\symt) ]$ are positive diagonal matrices. Using \eqref{mle_error_subcovariance}, we then write
\begin{align}\label{mle_error_subcovariance_diff}
    &\lt(\mt{Var}\!\lt[\matb \hat Y_\symi^{2,\Re} \\ \hat Y_\symi^{2,\Im} \mate\rt]\rt)^{-1} - \lt(\mt{Var}\!\lt[\matb \hat Y_\symi^{1,\Re} \\ \hat Y_\symi^{1,\Im} \mate\rt]\rt)^{-1} =\\ 
    &\quad \sum_{q,t=1}^{2,N}
    \lt(\matb \alpha^\Re \Re(\Phi_{\symi N+\symt})^\top \Re(\Phi_{\symi N+\symt}) &\hspace{-8px}
    0 \\ 0 &\hspace{-8px}
    \alpha^\Im \Im(\Phi_{\symi N+\symt})^\top \Im(\Phi_{\symi N+\symt}) \mate\rt) \hspace{-3px}, \nonumber
\end{align}
with $\alpha^\Re = D_{\Re,q\symi\symt}(Y_{\cdot\symi}^2)^{-1} - D_{\Re,q\symi\symt}(Y_{\cdot\symi}^1)^{-1}$ and $\alpha^\Im = D_{\Im,q\symi\symt}(Y_{\cdot\symi}^2)^{-1} - D_{\Im,pit}(Y_{\cdot\symi}^1)^{-1}$. The inequalities in \eqref{fisher_D_ineq} show that both $\alpha^\Re$ and $\alpha^\Im$ are non-positive. Hence, the blocks of \eqref{mle_error_subcovariance_diff} are the product of a negative scalar and a quadratic form and are therefore negative semi-definite. From this observation, it follows that
\begin{align}\label{mle_error_subcovariance_ineq}
    &\mt{Var}\!\lt[\matb \hat Y_\symi^{2,\Re} \\ \hat Y_\symi^{2,\Im} \mate\rt] \succeq \mt{Var}\!\lt[\matb \hat Y_\symi^{1,\Re} \\ \hat Y_\symi^{1,\Im} \mate\rt],
\end{align}
which finishes the proof.

\section{Non diagonal Bayesian prior}\label{appendix_nondiag}
The goal of the regularization \eqref{eq_prior_nondiag} of diagonal elements of $\hat Y$ is to reduce the bias from a sparsity-promoting prior. This means that the diagonal elements should keep a value close to the one estimated with MLE. For all $\symi$ and with the structural prior described in \cref{subsec_structural}, this means
\begin{align*}
    \hat Y_{\mt{MLE},\symi\symi} &= \frac{1}{2} \bb 1_{n^2}^\top \vect\lt( \begin{array}{c|c|c} \bb 0 & \hat Y_{\symi,:\symi-1}^\top & \bb 0 \\ \hline \vspace{-11px} \\
    \hat Y_{\symi,:\symi-1} & 0 & \hat Y_{\symi,\symi+1:} \\ \hline \vspace{-11px} \\
    \bb 0 & \hat Y_{\symi,\symi+1:}^\top & \bb 0 \end{array} \rt), \\
    &= [1, \imj] \otimes \lt(\ve(e_\symi \bb 1_n^\top + \bb 1_n e_\symi^\top)\rt)^\top \hat y_r,
\end{align*}
where $[\hat Y_{\symi,:\symi-1}, \hat Y_{\symi\symi}, \hat Y_{\symi,\symi+1:}] = \hat Y_\symi$ is the $\symi$th row of $\hat Y$. The non-diagonal prior is then given by $\mu_{\mt{nd}} = \bb 1_n$ and $L_{\mt{nd}}$ defined as:
\begin{align*}
    L_{\mt{nd}} &= \frac{\lambda'}{\lambda}\! \lt[\lt(\matb \Re(\hat Y_{\mt{MLE},\symi\symi}^{-1}) & 0 \\ \hspace{-8px} 0 & \hspace{-8px} \Im(\hat Y_{\mt{MLE},\symi\symi}^{-1}) \mate \rt) \!\otimes \ve(e_\symi \bb 1_n^\top + \bb 1_n e_\symi^\top)\rt]_{\symi=1}^n \hspace{-5px}.
\end{align*}

\section{Noise bias in Cartesian coordinates} \label{appendix_bias}
In \cref{subsec_noise_carac} we stated that the noise bias \cref{eq:average-true-bias} is negligible with realistic levels of accuracy. In order to show this, we adapt the procedure in \cite[Sec III.A]{lerro1993tracking}. By using the first-order Taylor expansion of \cref{eq:average-true-bias} about $\sigma_\delta = 0$, we get:
\begin{subequations}
    \begin{align}
        \E[\Delta c | \Tilde{v}, \Tilde{\theta}] &\simeq - (\Tilde{v}\sigma^2_\delta/2)\cos\Tilde{\theta},\\
        \E[\Delta d | \Tilde{v}, \Tilde{\theta}] &\simeq - (\Tilde{v}\sigma^2_\delta/2)\sin\Tilde{\theta}.
    \end{align}
\end{subequations}
Then, 
\begin{equation}
    \lVert [\E[\Delta c | \Tilde{v}, \Tilde{\theta}], \E[\Delta d | \Tilde{v}, \Tilde{\theta}]] \rVert = \Tilde{v}\sigma^2_\delta/2.
\end{equation}
Moreover, the minimum eigenvalue $\lambda_{\min}$ of the covariance matrix \eqref{eq:average-true-variance} is $\min(\sigma^2_\epsilon, \Tilde{v}^2\sigma^2_\delta)$: therefore, the minimum standard deviation in the covariance matrix is $\sigma_{\text{min}} = \sqrt{\lambda_\text{min}} = \sqrt{\min(\sigma^2_\epsilon, \Tilde{v}^2\sigma^2_\delta)}$.

The bias can be considered non-significant if $\lVert [\E[\Delta c | \Tilde{v}, \Tilde{\theta}], \E[\Delta d | \Tilde{v}, \Tilde{\theta}]] \rVert/\sigma_\text{min}$ is small. Adopting the per-unit system, and using realistic $\mu$PMU accuracy specifications (\cref{tab:pmu-error}), we obtain:
\begin{equation}
\label{eq:noise-bias-significance}
    \frac{\lVert [\E[\Delta c | \Tilde{v}, \Tilde{\theta}], \E[\Delta d | \Tilde{v}, \Tilde{\theta}]] \rVert}{\sigma_\text{min}} \simeq 8.72\cdot10^{-5}.
\end{equation}
Thus, the bias is four orders of magnitude smaller than the smallest standard deviation from the noise covariance matrix and therefore can be safely neglected.

\bibliographystyle{unsrt}
\bibliography{references}

\end{document}